\pdfoutput=1
\newif\ifFull
\Fullfalse

\documentclass[10pt]{article}
\advance \topmargin by -\headheight
\advance \topmargin by -\headsep
\evensidemargin \oddsidemargin
\usepackage{graphicx}
\usepackage{url}
\usepackage{epstopdf}
\usepackage{subfig}
\usepackage{boxedminipage}
\usepackage{amsmath}
\usepackage{amsthm}
\usepackage{amssymb}
\usepackage{algorithmic}
\usepackage{color}

\usepackage[lined,boxed]{algorithm2e}
\usepackage{caption}
\usepackage[utf8]{inputenc}
\usepackage{datatool}
\usepackage[xindy]{glossaries}
\usepackage[noblocks]{authblk}

\newcommand{\eat}[1]{}

\let\epsilon=\varepsilon
\newcommand{\E}{\mathrm{E}}
\newcommand{\dist}{\text{dist}}

\newtheorem{theorem}{Theorem}
\newtheorem{fact}[theorem]{Fact}

\newtheorem{lem}[theorem]{Lemma}

\newtheorem{corollary}[theorem]{Corollary}

\allowdisplaybreaks%%%%%

\title{Voronoi Choice Games}

\author[1]{Meena Boppana}
\author[1]{Rani Hod}
\author[1]{Michael Mitzenmacher}
\author[1]{Tom Morgan}
\affil[1]{Harvard University\thanks{Meena Boppana was supported in part by a PRISE summer research fellowship.
Rani Hod was supported by the Center of Mathematical Sciences and Applications at Harvard University.
Michael Mitzenmacher was supported in part by NSF grants CCF-1320321, CNS-1228598, IIS-0964473, and CCF-0915922;  part of his work was done while visiting Microsoft Research, New England.
Tom Morgan was supported in part by NSF grants  CCF-1320231 and CCF-0915922. }}

\date{}

\begin{document}

\maketitle

\begin{abstract}
We study novel variations of Voronoi games and associated random
processes that we call {\em Voronoi choice games}.  These games
provide a rich framework for studying questions regarding the power of
small numbers of choices in multi-player, competitive scenarios, and they
further lead to many interesting, non-trivial random processes that
appear worthy of study.  

As an example of the type of problem we study, suppose a group of $n$
miners (or players) are staking land claims through the following
process: each miner has $m$ associated points independently and
uniformly distributed on an underlying space (such as the the unit
circle, the unit square, or the unit torus), so the $k$th miner will
have associated points $p_{k1},p_{k2},\ldots,p_{km}$.  We generally
here think of $m$ as being a small constant, such as 2.  Each miner
chooses one of these points as the base point for their claim.  Each
miner obtains mining rights for the area of the square that is closest
to their chosen base; that is, they obtain the Voronoi cell
corresponding to their chosen point in the Voronoi diagram of the $n$
chosen points.  Each player's goal is simply to maximize the amount of
land under their control.  What can we say about the players' strategy
and the equilibria of such games?

In our main result, we derive bounds on the expected number of pure
Nash equilibria for a variation of the 1-dimensional game on the
circle where a player owns the arc starting from their point and
moving clockwise to the next point.  This result uses interesting
properties of random arc lengths on circles, and demonstrates the
challenges in analyzing these kinds of problems.  We also provide
several other related results.  In particular, for the 1-dimensional
game on the circle, we show that a pure Nash equilibrium always exists
when each player owns the part of the circle nearest to their point,
but it is NP-hard to determine whether a pure Nash equilibrium exists
in the variant when each player owns the arc starting from their point
clockwise to the next point.  This last result, in part, motivates our
examination of the random setting.

\end{abstract}

\section{Introduction}
Consider the following prototypical problem: a group of miners are
staking land claims.  The $k$th miner---or player---has $m$
associated points $p_{k1},p_{k2},\ldots,p_{km}$ in the unit torus
(which is the unit square with wraparound at the boundaries, providing
symmetry).  Each miner via some process will choose exactly one of
their $m$ points as the base for their claim.  The resulting $n$
points yield a Voronoi diagram, and each miner obtains their
corresponding Voronoi cell.  Each player's goal is simply to maximize
the amount of land under their control.  We wish to study player
behavior in this and similar games, focusing on equilibria.

As another application, political candidates can often be mapped
according to their political views into a small-dimensional space;
e.g., American candidates are often viewed as being points in a
two-dimensional space, measuring how liberal/conservative they are on
economic issues in one dimension and social issues on the other.
Suppose parties must choose a candidate simultaneously, and their
probability of winning is increasing in the area of the political space
closest to their point.  Again, the goal in this case is to maximize the
corresponding area in a Voronoi diagram.  

There are numerous variations one can construct from this setting.
Most naturally, if the players are (lazy) security guards instead of
miners, who have to patrol the area closest to their chosen base,
their goal might be to minimize the area under their purview.  Other
alternatives stem from variations such as whether player choices are
simultaneous or sequential, how the points for players are chosen, the
underlying metric space, the type of equilibrium sought, and the
utility function used to evaluate the final outcome.

However, the variations share the following fundamental features.  
There are $n$ players, with the $k$th player having $m_k$ associated
points in some metric space.  (We will focus on $m_k = m$ for a fixed $m$ for all players.)
Each player will have to choose to adopt one of their available
points.  A Voronoi diagram is then constructed, and each player is
then associated with the corresponding area in the diagram.  We
refer to this general setting as {\em Voronoi choice games}.  
We discuss below how Voronoi choice games differ from similar recent
work, but the key point is in the problems we study different players have
different available choices;  this asymmetry creates new problems and requires
distinct methods.  

We are particularly interested in the setting where each player's
points are chosen uniformly at random from the underlying space.
While uniform random points are not motivated by practice, the
framework leads to an interesting and, from the standpoint of
probabilistic analysis and geometry, very natural class of games.
Our work suggests many potential connections, to work on Voronoi diagrams
for random point sets, and to work on balanced allocations (or ``the power
of two choices''), where choice is used to improve load balancing.  
Moreover, looking at the setting of uniform random points gives us 
the opportunity to understand the nature of these games at a high level;
specifically, do most instances have no pure Nash equilibrium, or could they have
exponentially many possible pure Nash equilibria?

In general, however, we find that results for these types of problems
seem very challenging.  In our main result, we limit ourselves to the
setting where each player has $m$ associated points chosen uniformly
at random from the unit circle, and each player owns the arc starting
from their point clockwise to the next point -- that is, the distance
is unidirectional around the circle.  We derive bounds on the expected
number of pure Nash equilibria.  Even in this simple setting, our
result is quite technical, requiring a careful analysis based on
interesting properties of distributions of random arcs on a circle.
This appears, however, to be the ``easiest'' interesting version of the problem;
currently, higher-dimensional Voronoi diagrams are beyond our reach. 
However, our work suggests that further results are likely to involve interesting mathematics.

The random case of this specific version of the problem is also motivated 
by the following results.  We show it is NP-hard to determine whether a pure Nash
equilibrium exists when a player owns the arc starting from their
point clockwise to the next point for $m \geq 4$, nearly resolving the worst
case.  Further, for the different setting when a player owns an arc of the
circle corresponding to the standard Voronoi diagrams, that is a player
own all points nearest to their point, we show a Nash equilibrium always
exists (as long as all the possible choices for the players are distinct).  

While other similar Voronoi game models have been introduced previously, our primary
novelty is to introduce this natural type of asymmetric ``choice'' into these types of games.
We believe this addition provides a rich framework with many interesting combinatorial,
geometric, and game theoretic problems, as we describe throughout the
paper.  As such, we leave many natural open questions.

\subsection{Related Work}

The classical foundations for problems of this type can be found in
the work by Hotelling \cite{hotelling1929stability}, who studied the
setting of two vendors who had to determine where to place their
businesses along a line, corresponding to the main street in a town,
with the assumption of uniformly distributed customers who would walk
to the nearer vendor.  Hotelling games have been considered for
example in work on regret minimization and the price of anarchy, where
the model studied players choosing points on a general graph instead of
on the line as in the original model \cite{blum2008regret}.  Recent work 
has also shown that for a Hotelling game on a given graph, once there 
are sufficiently many players a pure Nash equilibrium always exists \cite{fournier2014hotelling}.
A useful survey on economic location-based models is provided by Gabszewicz and Thisse \cite{GT}.

Other variations of Voronoi games have appeared in the literature.
More recent work refers to these generally as {\em competitive
location games}; see for
example \cite{durr2007nash,mavronicolas2008voronoi,teramoto2006voronoi,kiyomi2011voronoi},
which discuss Voronoi games on graphs, for additional references.

Our setting appears different from previous work, in that it focuses
on players with {\em limited} sets of choices that {\em vary} among the
players.  Our starting point was aiming to build connections between
Voronoi games and random processes based on ``the power of
two choices'' \cite{azar1999balanced,richa2000power,byers2003simple,achlioptas2009explosive}.
While in our games, each player has a limited (typically constant) set
of distinct points to choose from, in previous work generally {\em all
players} could choose from {\em any point} in the universe of possible
choices.  In economic terms, in relation to the Hotelling model, our
work models that different businesses may have available a limited 
number of differing locations where they may establish
their business.  For example, businesses may have optioned the right
to set up a franchise at specific locations in advance, and must then
choose which location to actually build.  While they could know the
options available to other competing franchises, they may have to
decide where to build without knowing the choices made by competitors.
In other situations, it may be possible for franchises to move (at some
cost) to an alternative location.  We emphasize that our model is 
very different than previously studied symmetric versions of the game;
we do not recover earlier results, and earlier results do not appear
to apply once asymmetry is introduced.

\subsection{Models}
Before beginning, we explain the general class of games we are interested in.
We refer to the following as the {\em $k$-D Simultaneous Voronoi Game}.  
\begin{itemize}
\item  Each of the $n$ players has $m$ associated points from the $k$-dimensional 
unit torus $[0,1]^k$.  We assume that all players know about all of the possible
points that can be chosen by every player (it is a game of complete information).
\item  The $n$ players must simultaneously choose one of their $m$ associated points.  
\item  A Voronoi diagram is constructed for the $n$ chosen points, and each player receives utility equal to the volume of its point's Voronoi cell
in the maximization variation of the game.  (In the minimization version, the utility could be the negation of the corresponding volume.)  
\end{itemize}
We note that in the Appendix we describe {\em sequential} as opposed to {\em simultaneous}
variations of these problems.  In what follows, we consider only simultaneous
versions, and drop the word where the meaning is understood.

The easiest version to think about is the 1-D version; each player
chooses from $m$ points on the unit circle, and after their choice
they own an arc of the circle corresponding to all points closest to
their chosen point.  If each player tries to maximize their arc
length, then the utility of a player is the length of their arc.  (Or, if
each player tries to minimize their arc length, the
negation of the arc length is the utility.)  On the unit circle, there
is another variant that we refer to as the {\em One Way 1-D
Simultaneous Voronoi Game}, in which a player owns the arc starting
from their point and continuing in a clockwise direction until the
next chosen point.  Such a variation is quite natural in one dimension;  it
corresponds to assigning a ``direction'' to the unit circle.  This variation is chiefly motivated by our connections to the power-of-two choices.  In particular, it resembles the distributed hashing scheme of \cite{byers2003simple} in which peers correspond to points on a circle and keys are mapped to the closest peer in one direction along the circle.

Our contributions include highlighting differences
between the 1-D problem and the One Way 1-D problem, showing that in
this case a small difference in the model subtlety leads to large
differences in the behavior with respect to equilibria.  Indeed, as we
explain, we believe the One Way 1-D problem potentially offers more
insight into the behavior of the $k$-D Simultaneous Voronoi Game for $k\geq 2$
with respect to pure Nash equilibria.  

We focus on analyzing the equilibria of these games.  The most common
equilibrium to study is the Nash equilibrium \cite{nash1951non}, in
which each player has a random distribution on strategies such that no
player can improve their expected utility by changing their
distribution.  While Nash's results imply the Voronoi games
above all have Nash equilibria, we do not determine the complexity of
finding Nash equilibria for these games; this is left as an open
question.  We here focus on pure Nash equilibrium.  
A pure Nash equilibrium is a Nash equilibrium in which each player's
distribution has a support of size one.  In other words, each player
picks a single strategy to play and, given the other players'
strategies, no player can improve their utility by choosing a
different strategy.  Unlike the Nash and correlated equilibria, a pure
Nash equilibrium is not guaranteed to exist.

Pure Nash equilibria can be viewed as a setting where each player
can choose to switch to any of their adopted points at any time.  
The question is then what are the stable
states, where no player individually has the incentive to switch their
adopted point.  These stable states correspond to pure Nash
equilibria, and may not even exist.  A natural question is whether
simple local dynamics---such as myopic best response, where at 
each time step some subset of players decides whether or not to switch the point
it has adopted---reach a stable state quickly.
To motivate our study of the random case, we examine the computational complexity of
determining the existence of stable states in the 
1-D Simultaneous Voronoi Game and One Way 1-D
Simultaneous Voronoi Game.  For the former, we show (making use
of known techniques) that a pure Nash equilibrium always exists;  
for the latter, we show that determining whether a pure Nash equilibrium
exists is NP-complete.  

We then consider the existence of a pure Nash equilibrium for the
Randomized One Way 1-D Simultaneous Voronoi Game, where each players
possible choices for points are selected uniformly at random from the
unit circle.  Here we bound the expected number of pure Nash
equilibria, through a careful analysis based on properties of
distributions of random arcs on a circle.

We also note that we have some results for another type of
equilibrium, known as the correlated equilibrium.  Whereas Nash
equilibria have the players independently choosing their strategies, a
correlated equilibrium allows the players' random distributions to be
correlated (for example, by an external party).  The stability
requirement is then that given knowledge only of the overall
distribution of outcomes and their own randomly chosen strategy, a
player cannot improve their expected utility by deviating from their
given strategy distribution \cite{aumann1974subjectivity}.  Since a
Nash equilibrium is a special case of a correlated equilibrium, a
correlated equilibrium for the above games must exist.  We discuss the
computational complexity of finding a correlated equilibrium for the
$k$-D Simultaneous Voronoi Game in Section~\ref{correlated}.

We provide additional results in the appendices, including an
empirical investigation of the probability that myopic best response
will find a stable state in the Randomized One Way 1-D Simultaneous
Voronoi Game and Randomized $2$-D Simultaneous Voronoi Game, and
several related conjectures related to the Randomized One Way 1-D
Simultaneous Voronoi Game.

\section{Correlated Equilibria}

\label{correlated}
Our goal in this section is to show that, for the $k$-D Simultaneous Voronoi Game, correlated
equilibria can be found in polynomial time.  We present the proof for $m=2$, 
but our results easily extend for $m>2$.  We present the results for $k = 1,2,$ and 3.  
The results appear to extend to higher dimensions but the geometric details are technical;  we
note the time required to determine the correlated equilibrium appears to grow as $n^{O(k)}$.  
The results also apply to the One Way 1-D Simultaneous Voronoi Game.
\begin{theorem}
\label{thm:ceq}
For $k=1,2,$ and 3, and for a fixed $m$, there is a polynomial time algorithm for finding a correlated equilibrium in the $k$-D Simultaneous Voronoi Game.
\end{theorem}

We appeal to \cite{coreq} and \cite{jiang2013polynomial}, who present
polynomial time algorithms for finding a correlated equilibrium of
games \emph{polynomial type}.  (The running times for these algorithms are not specifically presented in the papers
and appear rather large, but are still polynomial.)
A game of polynomial type is one that can be represented in
polynomial space such that given each player's strategy, their
utilities can be computed in polynomial time.  The $k$-D Simultaneous
Voronoi Game is of polynomial type because it can be represented in $O(nmk)$
space by a list of each players point choices and, given the players'
strategies, the utilities can be found by computing the Voronoi
diagram of the chosen points.

\label{thm:jiang}
\begin{theorem}[Theorem 4.5, \cite{jiang2013polynomial}]
Given a game of polynomial type and a polynomial time algorithm for computing the expected utility of a player under any product distribution on strategies, there exists a polynomial time algorithm for finding a correlated equilibrium in that game.
\end{theorem}

%% This paragraph can go if space is needed.
Jiang et al. proved Theorem~\ref{thm:jiang} by constructing a linear
program with a variable for each of the $2^n$ possible strategy
profiles.  The LP's constraints are non-negativity, and the constraints
requiring that the variables form a correlated equilibrium.  They do
not, however, enforce that the variables sum to one, or even at most
one, and rather use the sum of these variables as the objective.  Thus, since
a correlated equilibrium is guaranteed to exist by Nash's Theorem, this
LP is unbounded and its dual is infeasible.  They then run the
ellipsoid algorithm for a polynomial number of steps on the dual LP
(this takes polynomial time, since the dual LP has only polynomially
many variables).  They argue that the intermediate steps of the
ellipsoid algorithm can be used to construct product distributions of
which there is a convex combination that is a valid correlated
equilibrium, and which can be found with a second linear program.

The second linear program's separation oracle requires as a subroutine a polynomial
time algorithm for computing the expected utility of a player given a
product distribution over the strategies.  (This requirement is
referred to in \cite{coreq} and \cite{jiang2013polynomial} as the {\em polynomial
expectation property}.)  In our case, we represent this product distribution
by letting $s(p_i)$ be the probability the point $p_i$ is chosen, for $i \in \{1,\ldots,2n\}$.  For a given point $p$, which is one of two
choices for a player $i$, we use $\bar{p}$ to refer to player
$i$'s other choice, so we must have $s(p)=1-s(\bar{p})$.  Each player
independently chooses a point according to the probabilities given by $s$.  Our
work is to demonstrate polynomial time algorithms for this subroutine,
which we do in Section~\ref{oned} and Section~\ref{twod}.  We note that it is
not immediate that such an algorithm should exist, even when each
player has only $m=2$ choices, as the number of possible configurations is
$m^n$.  Hence, we cannot simply sum over all configurations when
calculating the expectation.  In \cite{coreq}
it is noted that for certain congestion games, these expectations can be computed
using dynamic programming, essentially adding one player in at a time and
updating accordingly.  Our approach is similar in spirit, but requires taking
advantage of the underlying geometry.  

\subsection{The Algorithm for the 1-D Game}
\label{oned}

The 1-D version of the problem, while simpler, provides the intuition that helps us in higher dimensions.

\begin{lem}
\label{lem:1dexpectedutil}
Computing a player's expected utility under a product distribution on strategies in the 1-D Simultaneous Voronoi Game takes $O(n \log n)$ time.
\end{lem}

\begin{proof}
Computing the expected utility of a product distribution in the 1-D Simultaneous Voronoi Game is equivalent to finding the expected area of the Voronoi cell owned by a given point (conditioned on it being chosen) given the independent probabilities of each other player's points being chosen.  In the $1$-dimensional torus (a circle), for a given point $p$, let $A_p$ be the area of $p$'s Voronoi cell (assuming it is chosen), $D_1(p)$ be the distance to the first chosen point clockwise from $p$, and $D_2(p)$ be the distance to the first point chosen counterclockwise from $p$.  We have $\E[A_p]=\frac{1}{2}(\E[D_1(p)]+\E[D_2(p)])$ by linearity of expectations.

To compute $\E[D_1(p)]$, we start by sorting all the points.  (We need
only sort once for all points $p$.)  We then loop over all points $p'$
moving clockwise from $p$ -- excluding $\bar{p}$, which we know
will not be chosen if $p$ is chosen -- and compute the probability
that $p'$ is the closest chosen point clockwise from $p$.  The main
insight is that, as we go through the points, the first time we have a
point $q$ such that we have already seen $\bar{q}$, then we can
stop the computation, as one of these two points must be chosen.

Algorithm~\ref{alg1d} describes how to compute
$\E[D_1(p)]$, assuming the points are given in sorted order (clockwise
from $p$).  We track $D$, the running value that will equal
$\E[D_1(p)]$ on return, and $P$, the probability that the closest
point clockwise from $P$ is further than those examined so far.  As we
visit a point $p'$, $D$ increases by $s(p') \cdot P \cdot \dist(p,
p')$.  The value of $P$ decreases by a factor of $(1-s(p'))$, unless
$\bar{p}'$ has already been visited; then $P$ becomes 0 and we
may terminate the calculation.  Computing $\E[D_2(p)]$ is entirely
analogous.  
%{\hspace*{\fill}\rule{6pt}{6pt}\smallskip}
\end{proof}

\begin{algorithm}
{\bf Compute} $\E[D_1(p)]$
\caption{Computing a player's expected utility in 1-D.}
\label{alg1d}
\begin{algorithmic}
\STATE $D \gets 0$ // value of $\E[D_1(p)]$ computed so far
\STATE $P \gets 1$ // probability of closest point clockwise from $p$ further than seen so far
\FOR{each point $q$ except $p$ and $\bar{p}$, in ascending order of clockwise dist from $p$}
	\IF {$\bar{q}$ has been previously visited} \RETURN $D + P \cdot \dist(p, q)$ \ENDIF
	\STATE $D \gets D + s(q) \cdot P \cdot \dist(p, q)$
	\STATE $P \gets P \cdot (1 - s(q))$
\ENDFOR
\RETURN $D$
\end{algorithmic}
\end{algorithm}

Observe that Lemma~\ref{lem:1dexpectedutil} makes no special use of
the fact that the points lie in a torus, and in fact the Lemma also
applies when they lie in a line segment.  The generalization to the
case where $m > 2$ is straightforward.  One proceeds with the same
computation to find $\E[D_1(p)]$, except that one can stop and return
the value for $D$ only when one has reached all $m$ points associated
with a some player (in the case of a torus). 

\subsection{The Algorithm in Higher Dimensions}
\label{twod}

In this section we present a polynomial time
algorithm for computing the expected utility of a product distribution
over the chosen points in 2-D, and explain how it generalizes to three dimensions, and may generalize further to higher dimensions.  Intuitively, the challenge in
going beyond one dimension is that the interactions are more complicated;
there are more than just two neighboring cells to consider. 
For convenience, the algorithm we present is for the unit square;
it is readily modified for the unit torus.

\begin{lem}
\label{lem:2dexpectedutil}
Computing a player's expected utility under a product distribution on strategies in the 2-D Simultaneous Voronoi Game takes $O(n^3 \log n)$ time.
\end{lem}

\begin{proof}
We assume that the points are in general
position (that is, no three are collinear); this assumption is not mandatory, but simplifies the proof.
As in the one-dimensional case, we decompose the
area of the cell into a sum and then exploit linearity of
expectations.  Given our initial point set,  there are $O(n^2)$ possible 
vertices that can make up a vertex on $p$'s Voronoi cell, since each
possible vertex is the circumcenter of $p$ and two of the $2n-2$ 
points that the other players can choose from.  Choose one arbitrarily to be $v_1$, and let 
$V(p)=\{v_1,\ldots,v_t\}$ be the set of all possible vertices, in
sorted order according to the angle between $v_1$, $p$, and $v_i$
(measured in a counterclockwise direction from $v_1$).  
Let $\theta_i$ be the corresponding angle for $v_i$.

Let the random variable $A_p$ be the area of $p$'s Voronoi cell.  $A_p
= \sum_i T_i$ where $T_i$ is the area of $p$'s Voronoi cell
between angles $\theta_i$ and $\theta_{i+1}$.  Thus, $\E[A_p]
= \sum_i \E[T_i]$.  Figure~\ref{vg2d} provides an example of this
decomposition for one possible realization of a point's Voronoi cell.  Note that as depicted in this figure, a vertex may appear inside the Voronoi cell if neither of the points associated with that vertex are chosen by their corresponding players.  Observe that region $T_i$ is always a triangle since,
by definition, no other vertex of the cell creates an angle in the
range $(\theta_i, \theta_{i+1})$.  If we let 
$D_i$ be the distance from $p$ to the boundary of its Voronoi 
cell in the direction $\theta_i$, then
$$\E[T_i] = \E\left[\frac{D_i
D_{i+1} \sin(\theta_{i+1}-\theta_i)}{2}\right] = \frac{\E[D_i
D_{i+1}]\sin(\theta_{i+1}-\theta_i)}{2}.$$

\begin{figure}[h]
	\centering
		\includegraphics[width=.5\textwidth]{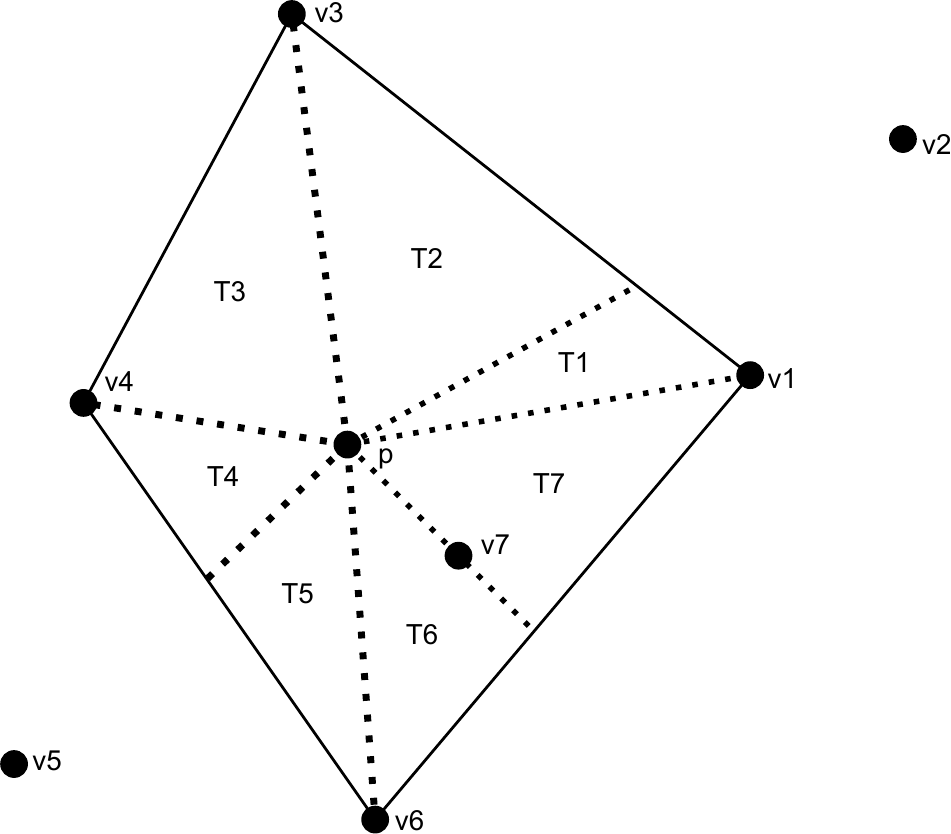}
	\caption{A point $p$ together with the possible vertices $V(p)$ and one possible Voronoi cell, decomposed into triangles $T_i$ by dotted lines.}
	\label{vg2d}
\end{figure}

We refer to the perpendicular bisector of the line $(p,q)$ as
the \emph{boundary line} between $p$ and $q$.  If $p$
and $q$ own neighboring Voronoi cells, this line separates
them.  The point $p$ has at most $2n$ possible boundary lines (one per
possible other point);  we say one of $p$'s boundary lines
is \emph{chosen} in the event that the corresponding $q$ is chosen by
its player.  Observe that the values of $D_i$ and $D_{i+1}$ are
uniquely determined by the chosen boundary line closest to $p$ within
the $(\theta_i, \theta_{i+1})$ region.  The notion of the closest
boundary line is unambiguous within $(\theta_i, \theta_{i+1})$ since,
by definition, none of the possible boundary lines intersect in
$(\theta_i, \theta_{i+1})$.

Algorithm~\ref{alg2d} provides an algorithm for computing $\E[D_i
D_{i+1}]$ given a product distribution on the points.  This algorithm
first sorts the possible boundary lines in increasing order of distance from
$p$ within $(\theta_i, \theta_{i+1})$.  This can be done by sorting them by their position along the ray $(p, (\theta_i+\theta_{i+1})/2)$.  The algorithm then iterates through the boundary lines
computing the probability that each one is the closest chosen boundary
line to $p$.  
The value  $D(p,q,i)$ in the algorithm is the distance from $p$ to the boundary line between $p$ and $q$ along angle $\theta_i$.
The running time is dominated by sorting the boundary
lines which takes $O(n \log n)$ time.  Thus, the total running time to
compute $\E[A_p]$ is $O(n^3 \log n)$. 
%{\hspace*{\fill}\rule{6pt}{6pt}\smallskip}
\end{proof}

\begin{algorithm}
{\bf Compute} $\E[D_i D_{i+1}]$
\caption{Computing part of a player's expected utility in 2-D.}
\label{alg2d}
\begin{algorithmic}
\STATE $D \gets 0$ // value of $\E[D_i D_{i+1}]$ computed so far
\STATE $P \gets 1$ // probability of closest boundary line to $p$ in this arc being further \\ than seen so far
\FOR{each point $q$ except $p$ and $\bar{p}$ whose boundary line with $p$ intersects \\ $(\theta_i, \theta_{i+1})$, in increasing order of the boundary line's distance to $p$}
	\IF {$\bar{q}$ has been previously visited} \RETURN $D + P \cdot D(p,q,i) \cdot D(p,q,i+1)$ \ENDIF
	\STATE $D \gets D + s(q) \cdot P \cdot D(p,q,i) \cdot D(p,q,i+1)$
	\STATE $P \gets P \cdot (1 - s(q))$
\ENDFOR
\RETURN $D$
\end{algorithmic}
\end{algorithm}

\begin{corollary} \label{cor:vg3d}
Computing a player's expected utility under a product distribution on strategies in the 3-D Simultaneous Voronoi Game takes $O(n^5 \log n)$ time.
\end{corollary}

\begin{proof}
The approach of Lemma~\ref{lem:2dexpectedutil} generalizes naturally to three dimensions.  The idea is to partition the space into pyramids radiating out from $p$, in which none of the $O(n)$ possible boundary planes intersect.  Given this partitioning, we can again use linearity of expectations by computing the expected volume of $p$'s Voronoi cell in each pyramid separately.  Within a given pyramid, we sort the boundary planes by their distance to $p$, and then iterate through them computing the probability that each one is the closest chosen boundary plane.  We then multiply these probabilities by the volumes of the polyhedra induced by the corresponding boundary planes, and sum the results to get the expected volume of the pyramid.

Identifying the space partition in three dimensions is trickier than it is in two.  Here, the pyramids can be found by computing the lines of intersection between all of the possible boundary planes, projecting them onto an infinitesimal sphere centered at $p$, and then finding all of the faces induced by these lines on the sphere.  Each face on the sphere induces a pyramid in our partition by extending out from $p$ through the face.  Identifying these faces can be done by constructing the graph formed by the lines and their intersections and then proceeding analogously to finding the faces of a planar graph as in \cite{nishizeki1988planar}.  This process of identifying the space partition is depicted in Figure~\ref{fig:vg3d}. This takes time linear in the number of edges, faces and vertices.

\begin{figure}[!ht]
    \subfloat[Compute all the intersection lines between the possible boundary planes.]{%
      \includegraphics[width=0.3\textwidth]{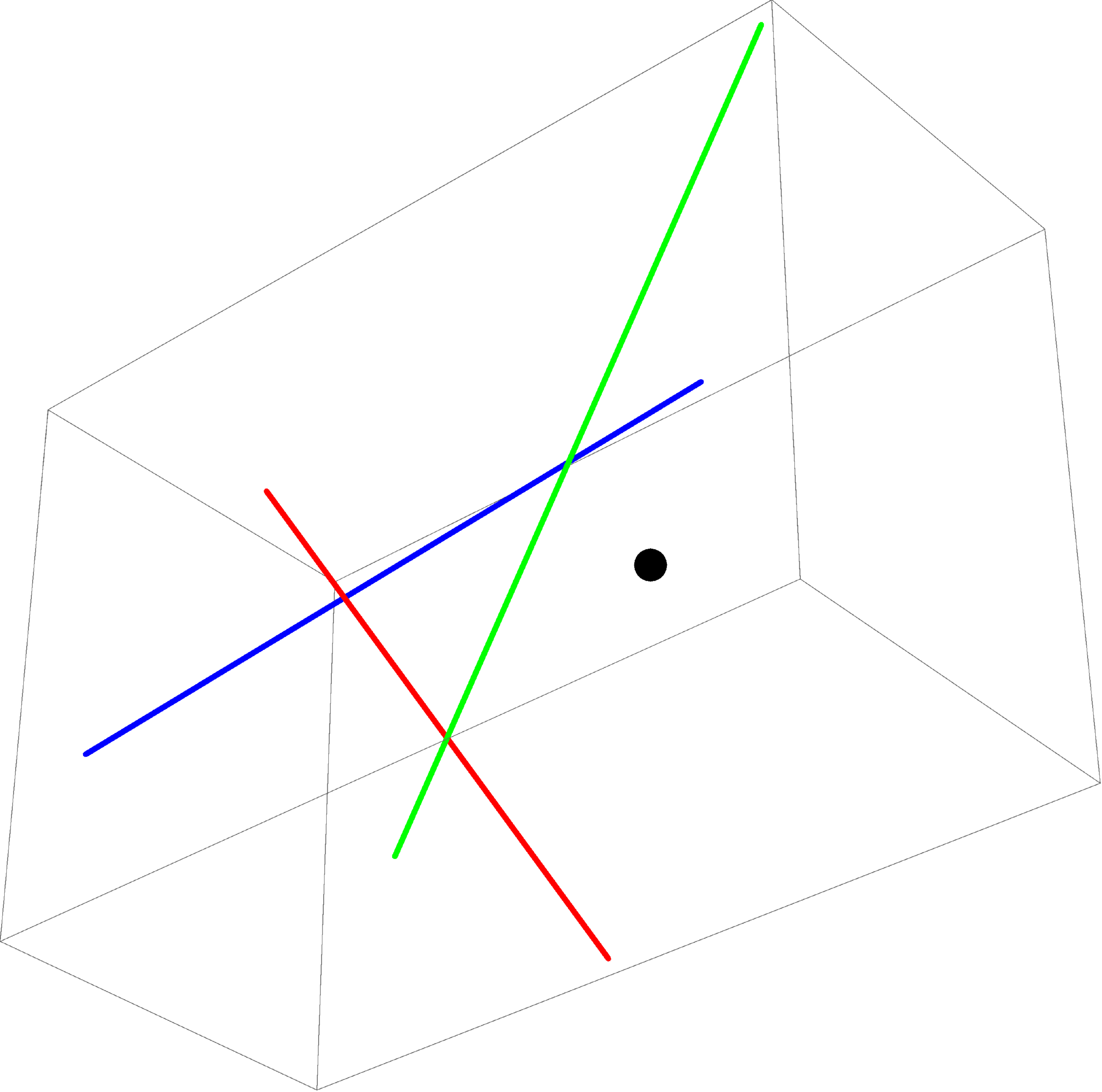}
    }
    \hfill
    \subfloat[Project the lines of intersection onto a sphere centered at our point.]{%
      \includegraphics[width=0.3\textwidth]{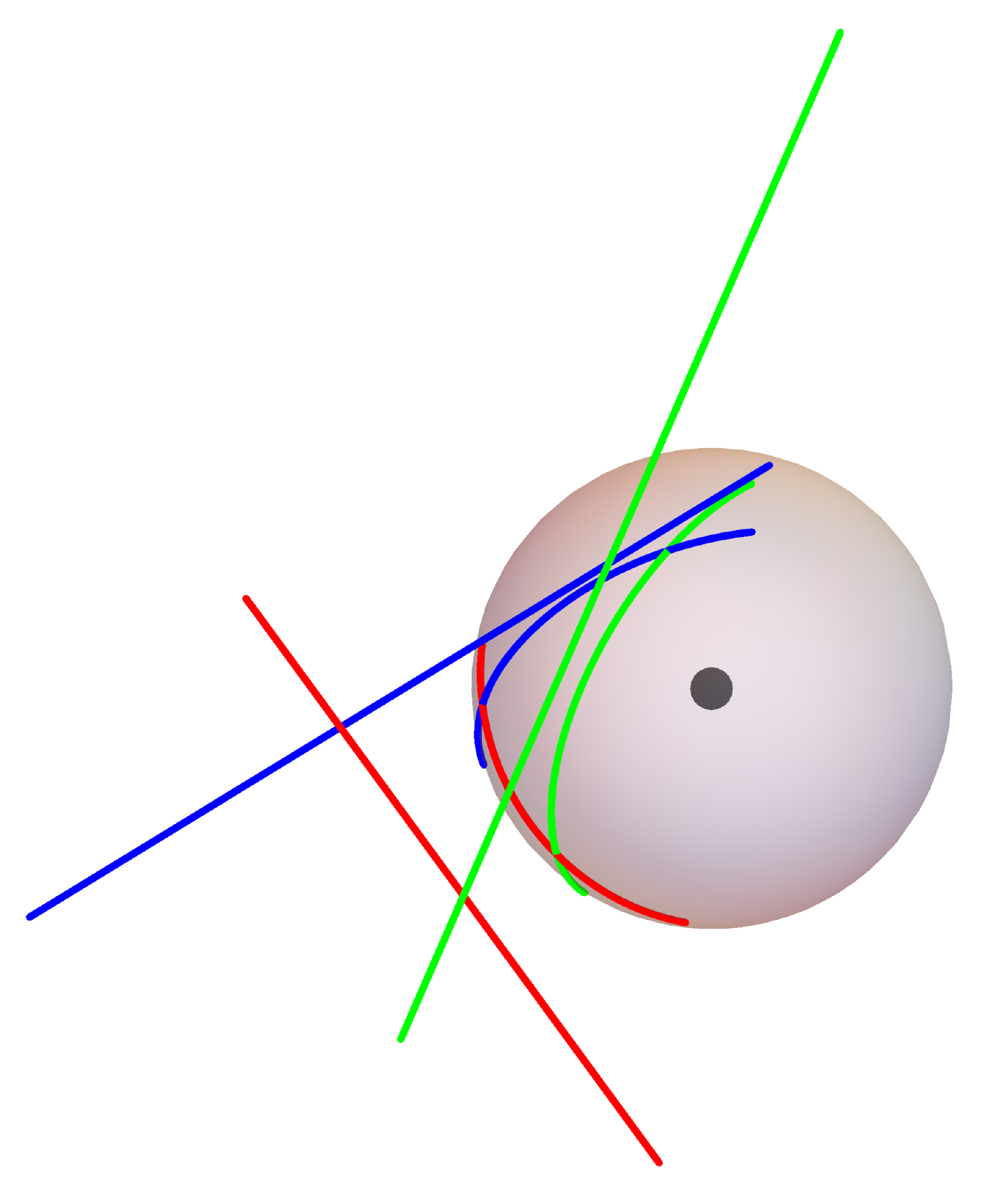}
    }
    \hfill
    \subfloat[Each region is formed by extending a pyramid from our point out through one of the faces of the graph induced by the lines on the sphere.]{%
      \includegraphics[width=0.3\textwidth]{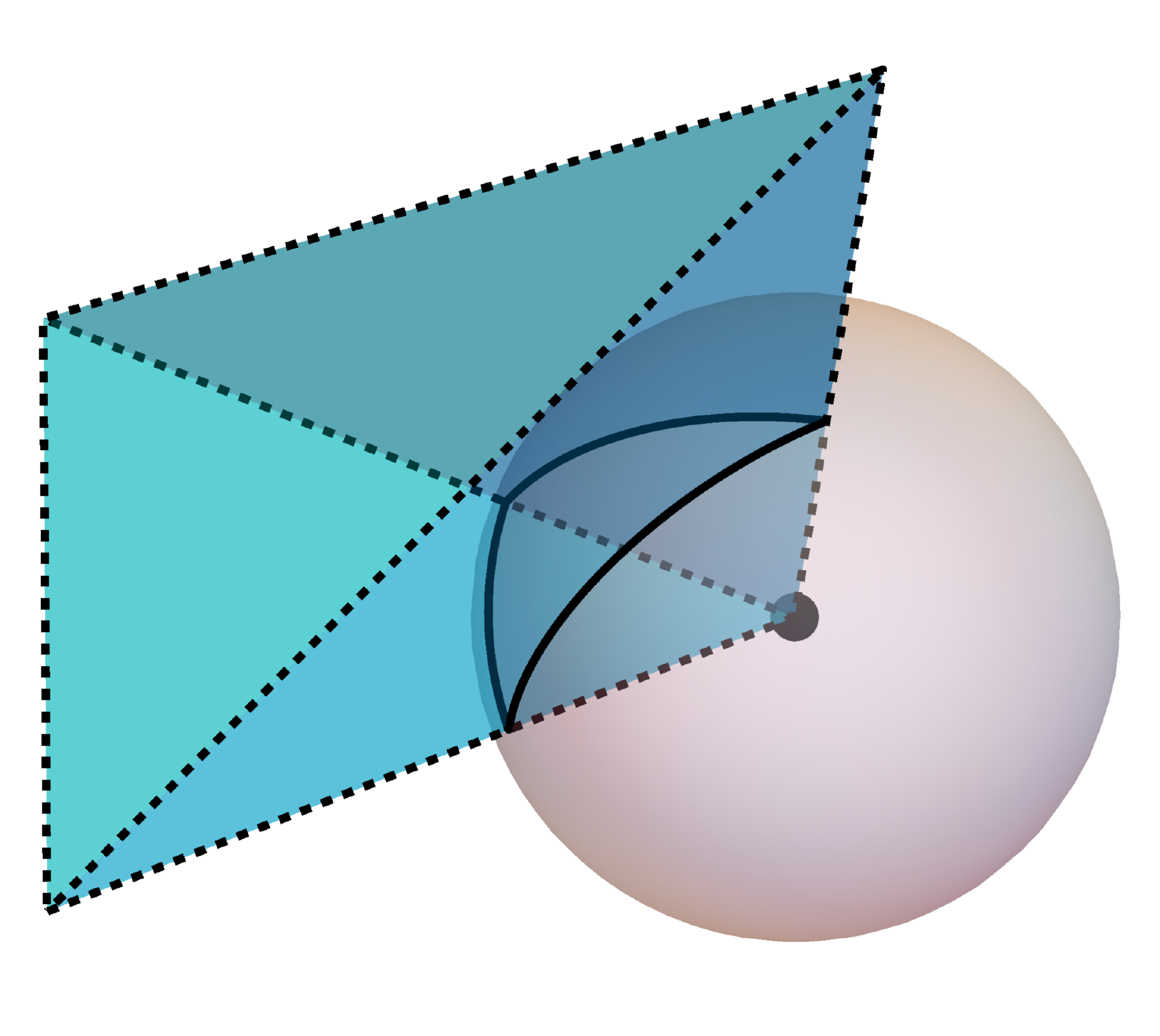}
    }
    \caption{The algorithm for computing the space partition required in Corollary~\ref{cor:vg3d}.}
    \label{fig:vg3d}
  \end{figure}

With $O(n^2)$ lines of intersection, there will be at most $O(n^4)$ vertices and faces.  For each pyramid, we must sort the $O(n)$ possible boundary planes, thus the running time for finding the expected utility in three dimensions is $O(n^5 \log n)$.
%{\hspace*{\fill}\rule{6pt}{6pt}\smallskip}
\end{proof}

This same strategy generalizes further to $k > 3$ dimensions, where we are now projecting $(k-2)$-dimensional hyperplanes formed by the intersection of the $(k-1)$-dimensional boundary hyperplanes onto a $k$-dimensional hypersphere.  However, determining methods for algorithmically identifying the ``faces'' on the hypersphere is outside the scope of this work.

\section{Pure Nash Equilibria}
\label{pure-ne}

In this section, we show a fundamental difference between the One Way
1-D Simultaneous Voronoi Game and the 1-D Simultaneous Voronoi
Game.  Recall that for these problems each of the $n$ players has a
choice of $m$ points on the unit circle; all players
simultaneously choose one of their $m$ points.  The utility for the
One Way variation of given player is equal to the distance to the
nearest chosen point clockwise from its chosen point, while for the
standard variation the utility is the size of the Voronoi cell (in
this case, an arc).

We show the standard Voronoi variation always has at least one pure
Nash equilibrium (for any number of choices per player), while it is
NP-hard to determine if the maximization version of the One Way variation has a pure Nash
equilibrium.  We also suggest the implications of these results for the 
higher dimensional setting.

\subsection{Existence of Pure Nash Equilibria in the 1-D Simultaneous Voronoi Game}

In the argument that follows we assume the choices of points are distinct.  The analysis can be easily modified
for the case where multiple players can choose the same point if ownership of that point is determined by a fixed preference order
(and other players have zero utility).  
However, if players choosing the same point share utility, then the theorem does not hold, as shown in \cite{mavronicolas2008voronoi}.
\begin{theorem}
\label{thm:simple}
A pure Nash equilibrium for the maximization and minimization versions 
of the 1-D Simultaneous Voronoi Game exists for any set of points.
\end{theorem}
\begin{proof}
We follow an approach utilized in \cite[Lemma 4]{durr2007nash}.  We
define a natural total ordering on multi-sets of numbers $A$ and $B$.
For any two such multi-sets $A$ and $B$, if $|A| < |B|$ we have
$A \succ B$.  When $|A| = |B|$, we have $A \succ B$ if $\max A > \max
B$.  If $\max A = \max B$, then let $A'$ be $A$ with one copy of the
value $\max A$ removed, and similarly for $B'$; then $A \succ B$ also
when $\max A = \max B$ and $A' \succ B'$.

Now consider a collection of choices in the 1-D Simultaneous Voronoi
Game, and let $A = \{a_0,a_1,\ldots,a_{n-1}\}$ be the corresponding
arc lengths, starting from some chosen point and then in clockwise
order around the circle, induced by the choice of points.  Each player's
payoff is given by a value of the form $(a_i + a_{i+1})/2$ for a suitable value of $i$.
(One player's payoff is $(a_{n-1}+a_0)/2$.)  Consider the maximization version of the game.
If some player has a move that improves their utility, let them make that move.
Without loss of generality, suppose this player's payoff was given by $(a_i + a_{i+1})/2$,
and it moves somewhere on the arc with length $a_j$.  Note this means $a_j > a_i + a_{i+1}$.
We see that the arc lengths $A'$ are those of $A$ but with $a_i, a_{i+1},$ and $a_j$ replaced by
$a_i+a_{i+1}, x,$ and $y$ where $x + y = a_j$.  Hence $A \succ A'$, so after a finite number
of moves, this version of myopic best response converges to a pure Nash equilibrium.  

The argument for the minimization variation is analogous.
%{\hspace*{\fill}\rule{6pt}{6pt}\smallskip}
\end{proof}

We note that we leave as an open question to determine a bound on the number of steps a myopic best response
approach would take to reach a pure Nash equilibrium;  in particular, we do not yet know if the 
approach of Theorem~\ref{thm:simple} yields a pure Nash equilibrium in a polynomial number of steps.

\begin{figure}
	\centering
		\includegraphics[width=.6\linewidth]{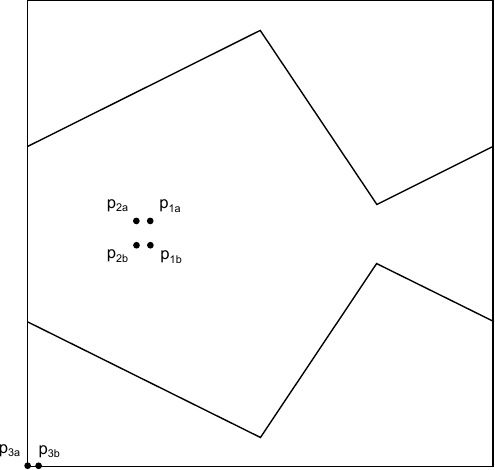}
	\caption{This is an example of a 2-D Simultaneous Voronoi Game in which no pure Nash equilibria exists.  Player 1 has points $(p_{1a}, p_{1b}) = ((1/4+\epsilon^2,1/2+\epsilon),(1/4+\epsilon^2,1/2-\epsilon))$, player 2 has points $(p_{2a}, p_{2b})=((1/4-\epsilon^2,1/2+\epsilon),(1/4-\epsilon^2,1/2-\epsilon))$, and player 3 has points $(p_{3a}, p_{3b}) = ((0, 0),(\epsilon, 0))$.  The Voronoi diagram shown is constructed from the points $(1/4,1/2)$ and $(0,0)$.}
	\label{fig:no_pne}
\end{figure}

We might have hoped that the above technique could allow us to show that for the 2-D 
Simultaneous Voronoi Game (and higher dimensions) that a pure Nash equilibrium exists.
Unfortunately, that is not the case.  One can readily find choices of 2 points for each of
3 players where no pure Nash equilibrium exists for both the maximization and minimization version
of the problem.  We have generated many such examples randomly, computing the Voronoi diagrams for
the eight resulting configurations.  One example in two dimensions is depicted in Figure~\ref{fig:no_pne}.  In this example, player 1 has choices $((1/4+\epsilon^2,1/2+\epsilon),(1/4+\epsilon^2,1/2-\epsilon))$, player 2 has choices $((1/4-\epsilon^2,1/2+\epsilon),(1/4-\epsilon^2,1/2-\epsilon))$, and player 3 has choices $((0, 0),(\epsilon, 0))$ for sufficiently small $\epsilon > 0$.  The idea here is that player 3's choice is irrelevant, and players 1 and 2 are dividing up the Voronoi cell owned by point $(1/4, 1/2)$ in the Voronoi diagram of the points $(1/4,1/2)$ and $(0,0)$.  If both player 1 and 2 choose their first point, or both choose their second point, then they divide the cell with a vertical line, which due to the geometry of the cell favors player 1.  However, if one of them chooses their first point and the other chooses their second point then they divide the cell with a roughly horizontal line, which gives each of them roughly half.  Thus, there is no choice of points for which neither wants to deviate.  This example applies
equally in higher dimensions (by making higher dimensional coordinates all zero).    

This example can be extended to show that for any number of players
$n$, there are settings of points for the players so no pure Nash equilibrium exists, showing that this
setting differs from previous work on symmetric Hotelling games on graphs, where it
has been shown that when there are sufficiently many players a pure Nash equilibrium always exists \cite{fournier2014hotelling}.
Specifically, use the same example but for players $4,\ldots,n$ both of their choices will be in an epsilon-small neighborhood of $(0,0)$.
We note that such examples do not disprove the possibility that a pure Nash equilibrium exists with high probability if the points are chosen randomly. 

Given that the 1-D Simultaneous Voronoi Game appears to have a very special structure (in terms of the existence of 
pure Nash equilibrium) that differs from $k$-D Simultaneous Voronoi Game for $k \geq 2$, it is natural to seek 
a 1-D variant that might shed more insight into the behavior in higher dimensions.  This motivates us to look at the 
One Way 1-D Simultaneous Voronoi Game.

\subsection{NP-Hardness of the One Way 1-D Simultaneous Voronoi Game}
\label{nphardness}

In contrast to the result of the previous subsection, we prove the NP-hardness of the maximization version of the One Way 1-D Simultaneous Voronoi Game whenever each player has $m \geq 4$ choices.  We leave as an open question to find a reduction for $m = 2$ or $3$, as well as for the minimization version.

In what follows we denote a player $r$'s $i$th point by $Z(r, i)$, as we have multiple different players we describe.

\begin{theorem} \label{thm:main}
The problem of determining if a pure Nash equilibrium (PNE) exists in the maximization version of a One Way 1-D Simultaneous Voronoi Game is NP-hard for $m \geq 4$.
\end{theorem}

Our proof will ultimately only use at most $4$ choices per player, however observe that we can essentially reduce the number of choices available to a player by placing leftover choices arbitrarily close in front of one of the desired choices, thus making the desired choices dominate the leftover choices.

Before we prove Theorem~\ref{thm:main}, we first introduce a gadget that will be
useful in our proof.  This is the Utility Enforcement Gadget (UEG),
depicted in Figure~\ref{fig:ueg}.  For now, in what follows $d$ and $\epsilon$
are suitable constants (we specify their values further in the proof
of the theorem).  The purpose of the UEG is to prevent a
PNE in which a given player $r$ has utility less than a given value
$d$.  The UEG consists of a region of length $d+7\epsilon$, where
$\epsilon$ is a value small enough that $d-\epsilon$ is greater than
any other utility $r$ can achieve less than $d$.  We also introduce
$5$ new players -- $b_1, b_2, x, y$ and $z$ -- and one new choice for
$r$: $Z(r, t+1)$, where $t$ is the number of other choices $r$ has.  Players $b_1$ and $b_2$ each only have one choice; these
form the boundaries of the UEG's region.  $Z(r, t+1)$ is $d-\epsilon$
before $Z(b_2, 1)$ with nothing in between them, so $r$ will always
prefer $Z(r, t+1)$ to any of its other choices with utility less than
$d$.  Players $x$ and $y$ each have two choices and $z$ has one choice.  $x,
y$ and $z$'s choices are arranged in the UEG's region such that they
have a stable configuration if and only if $r$ does not choose $Z(r,
t+1)$.

\begin{figure}
	\centering
		\includegraphics[width=.6\linewidth]{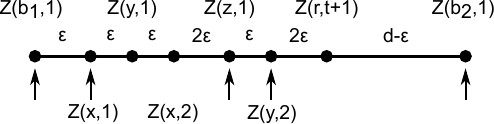}
	\caption{A Utility Enforcement Gadget for player $r$ and utility $d$.  The arrows point to the stable configuration for $b_1, b_2, x, y$ and $z$ when $r$ does not choose $Z(r,t+1)$.}
	\label{fig:ueg}
\end{figure}

Let $G$ be an instance of the One Way 1-D Simultaneous Voronoi Game, and let $G'$ be that instance modified to add a UEG (scaling as necessary) for a given player $r$ and utility $d$.

\begin{lem} \label{lem:ueg}
$G'$ has a PNE if and only if $G$ has a PNE in which $r$ has utility at least $d$.
\end{lem}
\begin{proof}
If $G$ has a PNE in which $r$ has utility at least $d$, then $G'$ has a PNE in which all of the players from $G$ make the same choices and $x$ chooses $Z(x,1)$ and $y$ chooses $Z(y,2)$.  Thus $r$ is the only player in $G$ whose preference is affected by the UEG, and if $r$'s utility is at least $d$, then $r$ will not wish to deviate to $Z(r,t+1)$ which only yields a utility of $d-\epsilon$.

If $G$ has only PNEs in which $r$ has utility less than $d$, then $G'$ will not have a PNE, since the preference of every player from in $G$ is unaffected by the addition of the UEG except for $r$.  In $G'$, $r$ will wish to deviate from any configuration in which $r$ has utility less than $d$ by choosing $Z(r, t+1)$.  However, there can be no PNE in which $r$ chooses $Z(r, t+1)$, because that will cause $y$ to prefer $Z(y,1)$ to $Z(y,2)$ when $x$ chooses $Z(x,1)$ and prefer $Z(y,2)$ to $Z(y,1)$ when $x$ chooses $Z(x,2)$, while $x$ prefers $Z(x,2)$ to $Z(x,1)$ when $y$ chooses $Z(y,1)$ and prefers $Z(x,1)$ to $Z(x,2)$ when $y$ chooses $Z(y,2)$.
%{\hspace*{\fill}\rule{6pt}{6pt}\smallskip}
\end{proof}

Now that we have defined the UEG, we can move on to the main proof.

\begin{proof}[Proof of Theorem~\ref{thm:main}]

We reduce from the NP-complete Monotone 1-in-3 SAT problem \cite{schaefer1978complexity}.  In this problem, we are given a set of $k$ variables $v_1,v_2,\ldots,v_k$ and $l$ clauses $c_1,c_2,\ldots,c_l$.  Each clause is a set of three variables, and the problem is to determine whether there is an assignment of boolean values to the variables such that in each clause, exactly one of its values is true.

\begin{figure}
	\centering
		\includegraphics[width=.5\linewidth]{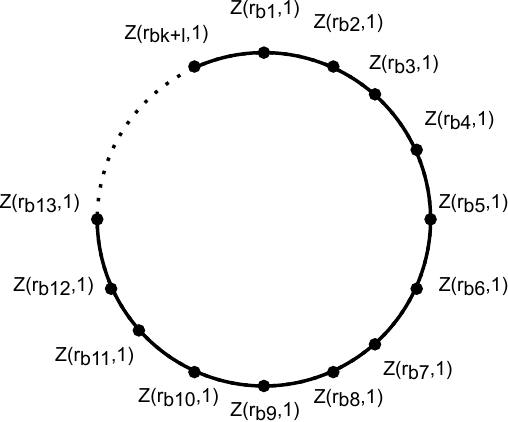}
	\caption{The partitioning of the circle using $k+l$ boundary players that each have only one point choice.}
	\label{fig:oneway_partition}
\end{figure}

\begin{figure}
	\centering
		\includegraphics[width=\linewidth]{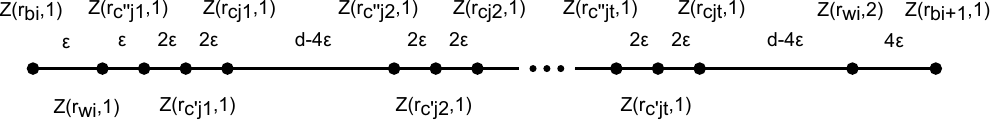}
	\caption{A variable region where for a variable contained in clauses $c_{j_1},\ldots,c_{j_t}$.  $\epsilon$ is much smaller than $d$.}
	\label{fig:oneway_varreg}
\end{figure}

Our reduction has $n = 2k+6l$ players.  The circle is
partitioned into $k+l$ equal sized regions using $k+l$ \emph{boundary
players} (denoted $r_{b_1},\ldots,r_{b_{k+l}}$) who each have only
one choice as shown in Figure~\ref{fig:oneway_partition}.  Within one of
these regions, we will sometimes use ``left'' to refer to
counterclockwise and ``right'' to refer to clockwise.  
We think of $l$ of the
players as corresponding to clauses (denoted $r_{c_1},\ldots,r_{c_l}$),
and refer to them as \emph{clause players}.  
We think of $k$ of the
regions as corresponding to variables.  Each clause player will have
one choice in each region that corresponds to one of that clause's
variables.  Which variable region the clause player chooses will
correspond to which of the clause's variables is set to true.  A
variable is true if every clause player corresponding to a clause
containing the variable chooses its point falling in the variable's
region.  A variable is false if none of the clause players
corresponding to clauses containing the variable choose points in the
variable's region.  Therefore, any intermediate configuration -- some
but not all of the relevant clause players choose points in a
variable's region -- is invalid.  To prevent such an invalid
configuration from yielding a PNE, we introduce an additional two
players per clause, that we will refer to as the \emph{shadow clause
players}, denoted by $r_{c_i'}$ and $r_{c_i''}$.  Each variable region
will also contain a \emph{wrap player} (denoted
$r_{w_1},\ldots,r_{w_k}$) who has one point choice on the left
of the region and one choice on the right side.  We place a wrap player
so that it prefers its left point unless one of the left most shadow
clause players' points are chosen, in which case it prefers its right
point.

In what follows, $d = 1$ and $\epsilon$ is a suitably small constant
much smaller than $d$.  To enforce a unit length for the circle, we simply
scale $d$ and $\epsilon$ appropriately.  A variable region will be
arranged as in Figure~\ref{fig:oneway_varreg}.  The clause players within
a variable region appear in the order of the clauses.  The clause
players are spaced $d$ apart, and each of the two shadow clause
players appear $2\epsilon$ and $4\epsilon$ before their corresponding
clause player.  The wrap player has one point $\epsilon$ before the
first shadow clause player, and one $4\epsilon$ before the terminating
boundary point.  This ensures that if neither of the left most shadow
clause players is chosen, then the wrap player will choose its left
point, and not affect any clause players.  However if one of the left
most shadow clause players is chosen, then the wrap player will choose
its right point thereby reducing the utility of the rightmost clause
player by $4\epsilon$.

Finally, we will have one UEG per clause player to enforce that its utility is
at least $d$.  These contribute the final $3l$ players (the boundaries for the
UEGs were included in the boundary players).

To complete the proof, we argue that any PNE in the game corresponds to a valid assignment of the variables and any valid assignment to the variables of corresponds to a PNE in the game.

Suppose that there exists a PNE such that there is a variable region with at least one clause player and at least one shadow clause player.  Then we must have that either there is a clause point followed immediately by a shadow clause point, or the first point is a shadow clause point and the last point is a clause point.  In the former case, we immediately have that the utility of the clause player is less than $d$, so such a PNE cannot exist by Lemma~\ref{lem:ueg}.  In the latter case, by the construction of the variable region, the wrap player in the region will prefer the point on the right of the region, where it will cause the utility of the clause player to once again be less than $d$.

Thus, any PNE has the property that each clause player is assigned to a variable region, and each variable region either has all clause players, or all shadow clause players.  This immediately gives us an assignment of variables that satisfies the 1-in-3 SAT formula, by setting each variable to be true if and only if there are clause players in the variable's region.

Suppose we have a variable assignment which satisfies the formula.  This corresponds to a PNE in the game.  Each clause player will choose the point in the region corresponding to the true variable in its clause.  So as long as all of the shadow clause players are assigned to other variable regions, none of the clause players will want to deviate since they will all achieve the maximum possible utility of $d$.  It remains to show that there is assignment of shadow clause players to remaining regions such that none of the shadow clause players want to deviate.  This can be achieved algorithmically.  First, pick an arbitrary assignment of the shadow clause players to the variable regions that are false.  Next, go through the shadow clause players in descending order of their corresponding clauses, and choose the best response for each one.  Because we ordered the clause players (and shadow clause players) in the regions by the order of the clauses, and the utility of a player is solely determined by the player to the right, no choice made by a player later in this procedure affects the utility of a player who chose earlier in this procedure.  Thus, at the end of the procedure all of the shadow clause players will still be in their best response.
%{\hspace*{\fill}\rule{6pt}{6pt}\smallskip}
\end{proof}

We conjecture that determining if a pure Nash equilibrium exists in the $k$-D Simultaneous Voronoi Game 
is NP-hard for $k \geq 2$;  however, we suspect that building the corresponding gadgets will prove technically
challenging.  

\section{Random Voronoi Games}

We now consider random variations of the Voronoi games we have
considered, where each player's available choices are chosen uniformly at
random from the underlying universe.  While it is not clear such a model corresponds to any
specific real-world scenario, such random problems are intrinsically
interesting combinatorially and in relation to other similar studied
problems.  For example, in the context of load-balancing in
distributed peer-to-peer systems, the authors
of \cite{byers2003simple} study a model where one begins with a
Voronoi diagram on $N$ points (chosen uniformly at random from the
universe, say the unit torus) corresponding to $N$ servers. Then $M$
agents sequentially enter; each is assigned $k$ random points from the
universe; and each agent chooses the one of its $k$ points that lies
in the Voronoi cell with the smallest number of agents, or {\em load}
for that server.  Other ``power-of-choice'' problems, such as the
Achlioptas process \cite{achlioptas2009explosive}, have spurred new
understanding of phenomena such as explosive percolation.  

Given our hardness results, a natural question is whether the One Way
1-D Simultaneous Voronoi Game has (or does not have) a pure Nash
equilibrium with high probability in the random variation.  While we
have not proven this result, we have proven bounds on the expected
number of pure Nash equilibria in the random setting that are
interesting in their own right and nearly answer this question.  In
particular, our careful analysis builds on the interesting
relationship between random arcs on a circle and weighted sums of
exponentially distributed random variables.

The following is our main result.
\begin{theorem} \label{thm:maxupperbound}
The expected number of PNE for the maximization version of the One Way 1-D Simultaneous Voronoi Game is at most $m$,
and at least $0.19^{m-1} m$.
\end{theorem}
Interestingly, our bounds depend on the number of choices $m$, not the number of players.  Unfortunately, this means
that we cannot use these expectation bounds directly to show that, for example, the probability of a PNE existing is
exponentially small in $n$.  But the bounds provide insight by showing that pure Nash equilibria are typically few
in number in the random case.  

In proving Theorem~\ref{thm:maxupperbound}, we need the following well-known property of exponential random variables.  Let $X_1,\ldots,X_n$ be i.i.d. exponential random variables.  Let $S_n = \sum_{j=1}^n X_j$.  Recall that $S_n$ is said to have a {\em gamma distribution}, and we use facts about the gamma distribution later in our analysis.
Similarly, $S_{n-1}/S_n$ is said to have a {\em beta distribution}, and we use facts about beta distributions as well.  For example,
\begin{lem} \label{lem:indep}
$S_n, \frac{S_1}{S_2}, \frac{S_2}{S_3}, \ldots, \frac{S_{n-1}}{S_n}$ are all mutually independent.
\end{lem}
See e.g. \cite{holst1980lengths} for its discussion of exponential random variables.  In particular, the $X_i$ can be viewed as the gaps in the arrivals of a Poisson process;  the $n-1$ arrivals before the last are uniformly distributed on the interval $[0,S_n]$, from which one can derive the lemma above.  

Another important fact we use is the following:
\begin{fact}
\label{fact}
We have $\E\left[\left(S_{j}/S_{j+1}\right)^{t}\right]=\prod_{i=0}^{t-1}\frac{j+i}{j+1+i}=\frac{j}{j+t}$,
since \\ $S_{j}/S_{j+1}\sim Beta\left(j,1\right)$.
\end{fact}

We note that in some of our arguments, the ordering of the variables
becomes reversed, and we consider sums of the form $\sum_{j=n-i+1}^n
X_j$.  Of course this has the same distribution as $\sum_{j=1}^i X_i$,
and the corresponding version of Lemma~\ref{lem:indep} holds.  Where
convenient, we therefore refer to $\sum_{j=n-i+1}^n X_j = S_i$ where there
is no ambiguity as to the desired meaning.

\begin{proof}[Proof of Theorem~\ref{thm:maxupperbound}]

We begin by using linearity of expectations to write the expected number of PNE in terms of the probability that each player choosing their first choice will yield a stable configuration.

\begin{align*}
\E[\text{\# PNE}] &= m^n \cdot \Pr[\text{first choices are stable}] \\
&= m^n \cdot \E[\Pr[\text{first choices are stable }|\text{ position of first choices}]].
\end{align*}

Partition the circle into arcs according to the players' first choice points.  Let $A_{i}$ be the length of the $i$th smallest arc.  As shown in \cite{holst1980lengths}, the $A_{i}$ are distributed jointly as
$$A_{i} \sim \frac{1}{S_n} \sum_{j=n-i+1}^n \frac{X_j}{j}$$
where again the $X_j$ are i.i.d. exponential random variables of mean 1 and $S_n = \sum_{j=1}^n X_j$.
We say that an arc is stable if the player whose point starts the arc (going clockwise) does not wish to deviate to any of their other points. 
Given the position of the first choices, the probability each arc is stable depends only on the other choices
available to the player that owns the arc, and hence the stability of the arcs are independent.  Therefore
\begin{align*}
\Pr&[\text{first choices are stable }|\text{ position of first choices}] \\
&= \prod_{i=1}^n{\Pr[i\text{th smallest arc is stable }|\text{ position of first choices}]}.
\end{align*}

If the arc is the $i$th smallest, then it will be stable if the other choices fall in the same arc, one of the $i-1$ smaller arcs,  or the front $A_i$-length portion of the $(n-i)$ larger arcs -- except in the latter two cases, we must take into account that if a choice falls immediately backward into the arc directly counterclockwise of the current arc, then the arc is not stable.  We therefore have the following calculation:  
\begin{align*}
\Pr&[i\text{th smallest arc is stable }|\text{ position of first choices}] \\
&= \left((n-i)A_i + \sum_{j=1}^i A_j - \min(A_i, \text{length of arc before i})\right)^{m-1} \\
&\leq \left((n-i)A_i + \sum_{j=1}^i A_j\right)^{m-1} \\
&= S_n^{-(m-1)} \left((n-i)\sum_{j=n-i+1}^n \frac{X_j}{j} + \sum_{j=1}^i \sum_{k=n-j+1}^n \frac{X_k}{k}\right)^{m-1} \\
&= S_n^{-(m-1)} \left((n-i)\sum_{j=n-i+1}^n \frac{X_j}{j} + \sum_{k=n-i+1}^n (k-n+i) \frac{X_k}{k}\right)^{m-1} \\
&= S_n^{-(m-1)} \left(\sum_{j=n-i+1}^n (n-i+j-n+i) \frac{X_j}{j}\right)^{m-1} \\
&= S_n^{-(m-1)} \left(\sum_{j=n-i+1}^n X_j\right)^{m-1} \\
&= \left(\frac{S_i}{S_n}\right)^{m-1}.
\end{align*}
Note we have used $\sum_{j=n-i+1}^n X_j = S_i$ for convenience.  Our resulting bound has a surprisingly clean
form in terms of the $S_i$.

Thus, by Lemma~\ref{lem:indep} and Fact~\ref{fact}
\begin{align*}
\Pr\left[\mbox{first choices are stable}\right] & \le\E\left[\prod_{i=1}^{n}\left(\frac{S_{i}}{S_{n}}\right)^{m-1}\right] \\
                                                & =\E\left[\prod_{i=1}^{n-1}\left(\frac{S_{i}}{S_{i+1}}\cdot\frac{S_{i+1}}{S_{i+2}}\cdot\cdots\cdot\frac{S_{n-1}}{S_{n}}\right)^{m-1}\right]\\
& =\E\left[\prod_{i=1}^{n-1}\left(\frac{S_{i}}{S_{i+1}}\right)^{i\left(m-1\right)}\right]=\prod_{i=1}^{n-1}\E\left[\left(\frac{S_{i}}{S_{i+1}}\right)^{i\left(m-1\right)}\right]\\
& =\prod_{i=1}^{n-1}\frac{i}{i+i\left(m-1\right)}=\frac{1}{m^{n-1}}.
\end{align*}
It then follows that
\begin{align*}
\E\left[\mbox{\#PNE}\right] & \le m^{n}\cdot\frac{1}{m^{n-1}}=m.
\end{align*}

We can similarly find a lower bound, although some additional technical work is required.

\begin{align*}
\Pr&[i\text{th smallest arc is stable }|\text{ position of first choices}] \\
&= \left((n-i)A_i + \sum_{j=1}^i A_j - \min(A_i, \text{length of arc before i})\right)^{m-1} \\
&\geq \left((n-i)A_i + \sum_{j=1}^i A_j - A_i\right)^{m-1} \\
&= \left((n-i-1)A_i + \sum_{j=1}^i A_j\right)^{m-1} \\
&= S_n^{-(m-1)} \left(\sum_{j=n-i+1}^n (n-i-1+j-n+i) \frac{X_j}{j}\right)^{m-1} \\
&= S_n^{-(m-1)} \left(\sum_{j=n-i+1}^n \frac{(j-1)X_j}{j}\right)^{m-1}
\end{align*}

Hence
\begin{align*}
\Pr[\text{first choices are stable}] &\geq \E\left[S_n^{-(m-1)n}\prod_{i=1}^n \left(\sum_{j=n-i+1}^n \frac{(j-1)X_j}{j}\right)^{m-1}\right].
\end{align*}
A simple stochastic domination argument (provided in Appendix B) shows that the expectation on the right side decreases if, in each term in the product, we 
equalize the coefficient, so that instead of terms of the form $\frac{(j-1)X_j}{j}$, the coefficient for all terms of the sum in the $i$th
term of the product is the average $c_i = \frac{1}{i} \sum_{j=n-i+1}^n \frac{j-1}{j}$.  This gives

\begin{align*}
\Pr[\text{first choices are stable}] &\geq  \E\left[S_n^{-(m-1)n}\prod_{i=1}^n \left(\sum_{j=n-i+1}^n c_i X_j \right)^{m-1}\right] \\
&= \E\left[S_n^{-(m-1)n}\prod_{i=1}^n \left(c_i S_i\right)^{m-1}\right] \\
&= \left( \prod_{i=1}^n c_i^{m-1} \right) \E\left[\prod_{i=1}^n \left(\frac{S_i}{S_n}\right)^{m-1}\right].
\end{align*}

Observe that this expectation is the same as the one we computed in the upper bound. Therefore
\begin{align*}
\E[\text{\# PNE}] &\geq \left(\prod_{i=1}^n c_i^{m-1}\right) m \\
&\geq 0.19^{m-1} m.
\end{align*}
The proof of the final inequality is presented in Appendix B.
%{\hspace*{\fill}\rule{6pt}{6pt}\smallskip}
\end{proof}

\eat{
{\bf TM: this is Rani's argument, which needs to be to formalized since that first $\approx$ is actually a $\leq$}
\begin{theorem} \label{thm:minupperbound}
The expected number of PNE for the minimization version of the One Way 1-D Simultaneous Voronoi Game is at most $TBD$.
\end{theorem}

\begin{theorem} \label{thm:minlowerbound}
The expected number of PNE for the minimization version of the One Way 1-D Simultaneous Voronoi Game is at least $\frac{m(m-1)}{(mn-1)n^{m-1}}$.
\end{theorem}
}

We note that similar calculations can be done for the minimization version, although we have not yet found a clean form
for the upper bound.  We can, however, state the following lower bound, showing the expected number of pure Nash equilibria is
at least inverse polynomial in $n$ for a fixed number of choices $m$.  
\begin{theorem} \label{thm:minlowerbound}
The expected number of PNE for the minimization version of the One Way 1-D Simultaneous Voronoi Game is at least $\frac{m(m-1)}{(mn-1)n^{m-1}}$.
\end{theorem}
\begin{proof}
This proof will proceed similarly to that of Theorem~\ref{thm:maxupperbound}.  We use the same notation, except that here $S_i = \sum_{j=1}^i X_j$.

An arc is stable if the player whose point starts the arc does not wish to deviate to any of their other points.  If the arc is the $i$th smallest, then it will be stable if the other choices fall in the preceding arc, or in a portion of length $A_j-A_i$  at the start of the $j$th smallest arc for $j > i$.  Therefore
\begin{align*}
\Pr&[i\text{th smallest arc is stable }|\text{ position of first choices}] \\
&= \left(\sum_{j=i+1}^n (A_j - A_i) + \min(A_i, \text{length of arc before i})\right)^{m-1} \\
&\geq \left(\sum_{j=i}^n (A_j - A_i) + A_1\right)^{m-1} \\
&= \left(\sum_{j=i}^n A_j - (n-i+1) A_i + A_1\right)^{m-1} \\
&= S_n^{-(m-1)}\left(\sum_{j=i}^n \sum_{k=n-j+1}^n \frac{X_k}{k} - (n-i+1) \sum_{k=n-i+1}^n \frac{X_k}{k} + \frac{X_n}{n}\right)^{m-1} \\
&= S_n^{-(m-1)}\left( \sum_{k=1}^{n-i} k \frac{X_k}{k} + \sum_{k=1}^n (n-i+1) \frac{X_k}{k} - (n-i+1) \sum_{k=n-i+1}^n \frac{X_k}{k} + \frac{X_n}{n}\right)^{m-1} \\
&= S_n^{-(m-1)}\left( \sum_{k=1}^{n-i} X_k + \frac{X_n}{n}\right)^{m-1} \\
&= S_n^{-(m-1)}\left( S_{n-i} + \frac{X_n}{n}\right)^{m-1}.
\end{align*}

Using the fact that probabilities of each arc being stable are independent when conditioning on their positions, we have
\begin{align*}
\Pr&[\text{first choices are stable}] \geq \E\left[S_n^{-(m-1)n}\prod_{i=1}^n \left(S_{n-i} + \frac{X_n}{n} \right)^{m-1}\right] \\
&\geq \E\left[S_n^{-(m-1)n} \left(\frac{X_n}{n}\right)^{m-1} \prod_{i=1}^{n-1} S_i^{m-1}\right] \\
&= \E\left[S_n^{-(m-1)n} X_n^{m-1} \prod_{i=1}^{n-1} S_i^{m-1}\right] / n^{m-1} \\
&= \E\left[S_n^{-(m-1)n} \left(\left(1-\frac{S_{n-1}}{S_n}\right)S_n\right)^{m-1} \prod_{i=1}^{n-1} S_i^{m-1}\right] / n^{m-1} \\
&= \E\left[S_n^{-(m-1)(n-1)} \left(1-\frac{S_{n-1}}{S_n}\right)^{m-1} \prod_{i=1}^{n-1} S_i^{m-1}\right] / n^{m-1} \\
&= \left(\E\left[\prod_{i=1}^{n-1} \left(\frac{S_i}{S_n}\right)^{m-1}\right] - 
\E\left[\left(\frac{S_{n-1}}{S_n}\right)^{m-1} \prod_{i=1}^{n-1} \left(\frac{S_i}{S_n}\right)^{m-1}\right]\right) / n^{m-1}.
\end{align*}

Observe here the occurrence of 
$$\E\left[\prod_{i=1}^{n-1} \left(\frac{S_i}{S_n}\right)^{m-1}\right] = \E\left[\prod_{i=1}^n \left(\frac{S_i}{S_n}\right)^{m-1}\right]$$ 
which we found to equal $m^{-(n-1)}$ in the proof of Theorem~\ref{thm:maxupperbound}.  Evaluating the other expectation using Lemma~\ref{lem:indep} and Fact~\ref{fact} we have

\begin{align*}
\E&\left[\left(\frac{S_{n-1}}{S_n}\right)^{m-1} \prod_{i=1}^{n-1} \left(\frac{S_i}{S_n}\right)^{m-1}\right] \\
&=\E\left[\left(\frac{S_{n-1}}{S_n}\right)^{m-1} \prod_{i=1}^{n-1}\left(\frac{S_{i}}{S_{i+1}}\cdot\frac{S_{i+1}}{S_{i+2}}\cdot\cdots\cdot\frac{S_{n-1}}{S_{n}}\right)^{m-1}\right] \\
&=\E\left[\left(\frac{S_{n-1}}{S_{n}}\right)^{n\left(m-1\right)}\prod_{i=1}^{n-2}\left(\frac{S_{i}}{S_{i+1}}\right)^{i\left(m-1\right)}\right] \\
&=\E\left[\left(\frac{S_{n-1}}{S_{n}}\right)^{n\left(m-1\right)}\right] \prod_{i=1}^{n-2}\E\left[\left(\frac{S_{i}}{S_{i+1}}\right)^{i\left(m-1\right)}\right] \\
&=\frac{n-1}{n-1+n(m-1)}\prod_{i=1}^{n-2}\frac{i}{i+i\left(m-1\right)} \\
&=\frac{n-1}{(mn-1)m^{n-2}}.
\end{align*}

Putting it all together we have

\begin{align*}
\E[\text{\# PNE}] &= m^n \Pr[\text{first choices are stable}] \\
&\geq m^n \left(m^{-(n-1)} - \frac{n-1}{(mn-1)m^{n-2}} \right) / n^{m-1} \\
&= m \left(1 - \frac{m(n-1)}{mn-1} \right) / n^{m-1} \\
&= m \left(\frac{m-1}{mn-1} \right) / n^{m-1} \\
&= \frac{m(m-1)}{(mn-1)n^{m-1}}.
\end{align*}
%{\hspace*{\fill}\rule{6pt}{6pt}\smallskip}
\end{proof}

\subsection{Empirical Results}

We have done several experiments regarding Random Voronoi games, which prove consistent with our theoretical results and
suggest some interesting conjectures, particularly for the 2-D Simultaneous Voronoi Game.  Our discussion of these results
can be found in Appendix C.

\section{Conclusion}
We have introduced a new but we believe important set of variants on
Voronoi games, where each player has a limited number of points to
choose from.  We believe these variations are motivated both by
natural economic settings, and because of the possible connections to
other ``power-of-choice'' processes in which participants choose from a limited set of random options.  

In particular, we note that the Voronoi choice games we propose offer
the chance to consider randomized versions of the problem, where the
set of possible choices for each player is chosen uniformly over of
the space.  We have conjectured that the Random $k$-D Simultaneous
Voronoi Game has a pure Nash equilibrium with probability approaching 1 as $n$ grows,
based on a simulation study.  While this is perhaps the most natural
open question in this setting, there remain several other questions
for both the simultaneous and sequential versions of Voronoi choice
problems, in the worst case and with random point sets.

%\bibliographystyle{plain}
%\bibliography{voronoi_game}

\begin{thebibliography}{10}

\bibitem{achlioptas2009explosive}
Dimitris Achlioptas, Raissa~M D'Souza, and Joel Spencer.
\newblock Explosive percolation in random networks.
\newblock {\em Science}, 323(5920):1453--1455, 2009.

\bibitem{aumann1974subjectivity}
Robert~J Aumann.
\newblock Subjectivity and correlation in randomized strategies.
\newblock {\em Journal of Mathematical Economics}, 1(1):67--96, 1974.

\bibitem{azar1999balanced}
Yossi Azar, Andrei~Z Broder, Anna~R Karlin, and Eli Upfal.
\newblock Balanced allocations.
\newblock {\em SIAM journal on computing}, 29(1):180--200, 1999.

\bibitem{blum2008regret}
Avrim Blum, MohammadTaghi Hajiaghayi, Katrina Ligett, and Aaron Roth.
\newblock Regret minimization and the price of total anarchy.
\newblock In {\em Proceedings of the Fortieth Annual ACM Symposium on Theory of
  Computing}, pages 373--382. ACM, 2008.

\bibitem{byers2003simple}
John Byers, Jeffrey Considine, and Michael Mitzenmacher.
\newblock Simple load balancing for distributed hash tables.
\newblock In {\em Peer-to-peer Systems II}, pages 80--87. Springer, 2003.

\bibitem{durr2007nash}
Christoph D{\"u}rr and Nguyen~Kim Thang.
\newblock Nash equilibria in voronoi games on graphs.
\newblock In {\em Proceedings of the 5th Annual European Symposium on
  Algorithms}, pages 17--28. Springer, 2007.

\bibitem{fournier2014hotelling}
Ga{\"e}tan Fournier and Marco Scarsini.
\newblock Hotelling games on networks: efficiency of equilibria.
\newblock {\em Available at SSRN 2423345}, 2014.

\bibitem{GT}
J.~J. Gabszewicz and J.-F. Thisse.
\newblock Location.
\newblock In {\em Handbook of Game Theory with Economic Applications}. Volume
  1, Chapter 9. R. Aumann and S. Hart, editors. Elsevier Science Publishers,
  1992.

\bibitem{holst1980lengths}
Lars Holst.
\newblock On the lengths of the pieces of a stick broken at random.
\newblock {\em Journal of Applied Probability}, pages 623--634, 1980.

\bibitem{hotelling1929stability}
Harold Hotelling.
\newblock Stability in competition.
\newblock {\em The Economic Journal}, 39(153):41--57, 1929.

\bibitem{jiang2013polynomial}
Albert~Xin Jiang and Kevin Leyton-Brown.
\newblock Polynomial-time computation of exact correlated equilibrium in
  compact games.
\newblock {\em Games and Economic Behavior}, 21(1-2):183--202, 2013.

\bibitem{kiyomi2011voronoi}
Masashi Kiyomi, Toshiki Saitoh, and Ryuhei Uehara.
\newblock Voronoi game on a path.
\newblock {\em IEICE TRANSACTIONS on Information and Systems},
  94(6):1185--1189, 2011.

\bibitem{mavronicolas2008voronoi}
Marios Mavronicolas, Burkhard Monien, Vicky~G Papadopoulou, and Florian
  Schoppmann.
\newblock Voronoi games on cycle graphs.
\newblock In {\em Proceedings of Mathematical Foundations of Computer Science},
  pages 503--514. Springer, 2008.

\bibitem{richa2000power}
M~Mitzenmacher, Andrea~W Richa, and R~Sitaraman.
\newblock The power of two random choices: A survey of techniques and results.
\newblock In {\em Handbook of Randomized Computing}, pages 255--312, 2000.

\bibitem{nash1951non}
John Nash.
\newblock Non-cooperative games.
\newblock {\em The Annals of Mathematics}, 54(2):286--295, 1951.

\bibitem{nishizeki1988planar}
Takao Nishizeki and Norishige Chiba.
\newblock {\em Planar graphs: Theory and algorithms}.
\newblock Elsevier, 1988.

\bibitem{coreq}
Christos~H Papadimitriou and Tim Roughgarden.
\newblock Computing correlated equilibria in multi-player games.
\newblock {\em Journal of the ACM}, 55(3):14, 2008.

\bibitem{schaefer1978complexity}
Thomas~J Schaefer.
\newblock The complexity of satisfiability problems.
\newblock In {\em Proceedings of the tenth annual ACM symposium on Theory of
  computing}, pages 216--226. ACM, 1978.

\bibitem{teramoto2006voronoi}
Sachio Teramoto, Erik~D Demaine, and Ryuhei Uehara.
\newblock Voronoi game on graphs and its complexity.
\newblock In {\em IEEE Symposium on Computational Intelligence and Games},
  pages 265--271. IEEE, 2006.

\end{thebibliography}

\newpage
\appendix

\section{Some Words on Models} \label{app:models}

In this paper we focus on games where players choose their
strategy simultaneously.  While possibly less interesting game theoretically, 
sequential variations where players arrive one at a time and decide their
choice with knowledge of the preceding choices of previous players are interesting.
For example, we might consider the {\em Random 2-D Sequential Voronoi Problem}:
\begin{itemize}
\item  The $n$ players arrive sequentially.  
Each player has $m$ associated points that are chosen independently and
uniformly at random from the unit torus $[0,1]^2$.  
\item  On arrival, each player chooses a point. The player's final payoff
will be the area of their cell in the Voronoi diagram.  The player's goal
is to choose the point that maximizes (or, alternatively, minimizes) their
final area.
\end{itemize}
In this variation, the primary natural question would be to determine
how each successive player calculates its optimal strategy, and in
particular under what conditions a player should not greedily choose
the point that yields the maximum area at the point of time of their
arrival.  Other natural questions would involve the distribution of
the sizes of the resulting cells (under the greedy strategy or some
other strategy), and how they differ from that of a Voronoi diagram of
$n$ uniform random points.  The ``unfairness'' in this setting might be measured by
the difference in the minimum and maximum cell areas, or the ``second
moment'' $\sum_i X_i^2$ where $X_i$ is the area obtained by the $i$th
player;  one could consider the unfairness compared to the number of choices.

\section{Random Voronoi Games Proofs} \label{sec:randomproofs}

We here show a statement made in proving the second part of Theorem~\ref{thm:maxupperbound}, 
that 
$$
E\left[\prod_{i=1}^n \left(\sum_{j=n-i+1}^n \frac{(j-1)X_j}{j}\right)^{m-1}\right]
\geq \E\left[\prod_{i=1}^n \left(c_i \sum_{j=n-i+1}^n X_j \right)^{m-1}\right],
$$
where $c_i = \frac1i\sum_{j=n-i+1}^n\frac{j-1}j$.
We actually prove a more general statement.  
Denote the product of the components of a vector $\vec{x}\in\mathbb{R}^{m}$
by $\Pi\left(\vec{x}\right)=\prod_{j=1}^{m}x_{j}$.

\begin{lem}
Let $A$ be an $n\times m$ matrix with non-negative entries, let
$s,t\in\left[m\right]$ such that $s\neq t$, and let $\epsilon>0$.
Let $X_{j}\sim Exp\left(1\right)$ i.i.d. for $j\in\left[m\right]$
and let $\vec{X}=\left(X_{1},\ldots,X_{m}\right)$. If $a_{is}\le a_{it}$
for all $i\in\left[n\right]$ then
$$
\E\left[\Pi\left(A\vec{X}\right)\right]\ge\E\left[\Pi\left(B\vec{X}\right)\right],
$$
where the perturbation $B$ is defined by 
$$b_{ij}=\begin{cases}
a_{ij}+\epsilon, & i=1\mbox{ and }j=s\\
a_{ij}-\epsilon, & i=1\mbox{ and }j=t\\
a_{ij}, & \mbox{otherwise.}
\end{cases}$$
\end{lem}
Repeatedly applying this allows us to average all non-zero entries in each row.

The intuition of the proof is rather simple; if we expand the products
out and use linearity of expectation, we can match up the terms in a
way that coordinates the powers of $X_s$ and $X_t$, controlling their moments. Because of this, on a term by
term basis, shifting weight from $a_{1t}$ to $a_{1s}$ can only reduce the value.

The following mechanics are somewhat unwieldy but formalize this idea.

\begin{proof}
We can write $\Pi\left(A\vec{X}\right)$ as a sum of $m^{n}$ terms
as follows:
$$
\Pi\left(A\vec{X}\right)=\sum_{f:\left[n\right]\to\left[m\right]}\prod_{i=1}^{n}a_{if\left(i\right)}X_{f\left(i\right)}
$$

Denote the number of preimages of $j\in\left[m\right]$ in a function
$f:\left[n\right]\to\left[m\right]$ by $\sigma_{f}\left(j\right)=\left|f^{-1}\left(j\right)\right|$.
By linearity of expectation and independence we get that
\begin{align*}
\E\left[\Pi\left(A\vec{X}\right)\right]
&=\sum_{f:[n]\to[m]}\E\left[\prod_{i=1}^n a_{if(i)}X_{f(i)}\right] \\
&=\sum_{f:[n]\to[m]}\prod_{i=1}^n a_{if(i))}\prod_{j=1}^m\E\left[X_j^{\sigma_f(j)}\right] \\ 
&=\sum_{f:[n]\to[m]}\left(\prod_{i=1}^n a_{if(i)}\prod_{j=1}^m\sigma_f(j)!\right).
\end{align*}
For notational convenience write
$$
P_{f}=\prod_{i=2}^{n}a_{if\left(i\right)}\prod_{j=1}^{m}\sigma_{f}\left(j\right)!
$$
and then we can express the expectation of $\Pi\left(B\vec{X}\right)$ as:
$$
\E\left[\Pi\left(B\vec{X}\right)\right]=\sum_{f:\left[n\right]\to\left[m\right]}\left(\prod_{i=1}^{n}b_{if\left(i\right)}\prod_{j=1}^{m}\sigma_{f}\left(j\right)!\right)=\sum_{f:\left[n\right]\to\left[m\right]}b_{1f\left(1\right)}P_{f}
$$
whereas their difference is:
\begin{align*}
\E\left[\Pi\left(A\vec{X}\right)-\Pi\left(B\vec{X}\right)\right] 
&=\sum_{f:\left[n\right]\to\left[m\right]}\left(a_{1f\left(1\right)}-b_{1f\left(1\right)}\right)P_f \\
&=\epsilon\left(\sum_{f\left(1\right)=t}P_{f}-\sum_{f\left(1\right)=s}P_{f}\right),
\end{align*}
so it suffices to show that $\sum_{f\left(1\right)=t}P_{f}\ge\sum_{f\left(1\right)=s}P_{f}$.

We now partition the sum to groups $\left\{ f,f',g,g'\right\} $ of
size four as follows. Take $f:\left[n\right]\to\left[m\right]$ such
that $f\left(1\right)=t$ and $\sigma_{f}\left(t\right)\le\sigma_{f}\left(s\right)+1$.
Let $S=f^{-1}\left(s\right)$ and $T=f^{-1}\left(t\right)\setminus\left\{ 1\right\} $.
Use the bijection from Lemma~\ref{lem:monotone-bijection} to map
$T\subset S\cup T$ to some $T'=\phi\left(T\right)\subset S\cup T$
of size $\sigma_{f}\left(s\right)$ such that $T\subset T'$. Define
$f',g,g':\left[n\right]\to\left[m\right]$ by
$$
f'\left(i\right)=\begin{cases}
t, & i\in T';\\
f\left(i\right), & \mbox{otherwise}
\end{cases}
$$
$$
g\left(i\right)=\begin{cases}
s, & i=1;\\
f\left(i\right), & \mbox{otherwise}
\end{cases}
$$
$$
g'\left(i\right)=\begin{cases}
s, & i=1;\\
g\left(i\right), & \mbox{otherwise}.
\end{cases}
$$
Observe that $\sigma_{g}\left(s\right)=\sigma_{f'}\left(t\right)=\sigma_{f}\left(s\right)+1$,
$\sigma_{g}\left(t\right)=\sigma_{f'}\left(s\right)=\sigma_{f}\left(t\right)-1$,
$\sigma_{g'}\left(s\right)=\sigma_{f'}\left(s\right)+1=\sigma_{f}\left(t\right)$
and $\sigma_{g'}\left(s\right)=\sigma_{f'}\left(s\right)+1=\sigma_{f}\left(t\right)$.

We claim that $P_{f}+P_{f'}\ge P_{g}+P_{g'}$. To see that, set aside
the common non-negative factor
$$
C=\prod_{i\in\left[n\right]\setminus\left(S\cup T\cup\left\{ 1\right\} \right)}a_{if\left(i\right)}\prod_{j=1}^{m}\sigma_{f}\left(j\right)!
$$
and observe that
\begin{eqnarray*}
P_{f} & = & C\cdot\prod_{i\in S}a_{is}\prod_{i\in T}a_{it}\\
P_{g'} & = & C\cdot\prod_{i\in S\setminus T'}a_{is}\prod_{i\in T'}a_{it}\\
P_{g} & = & C\cdot\frac{\sigma_{f}\left(s\right)+1}{\sigma_{f}\left(t\right)}\cdot\prod_{i\in S}a_{is}\prod_{i\in T}a_{it}\\
P_{f'} & = & C\cdot\frac{\sigma_{f}\left(s\right)+1}{\sigma_{f}\left(t\right)}\cdot\prod_{i\in S\setminus T'}a_{is}\prod_{i\in T'}a_{it}.
\end{eqnarray*}
Therefore
\begin{align*}
(&P_{f'}-P_{g})-(P_{g'}-P_{f}) \\
&=C\left(\frac{\sigma_{f}\left(s\right)+1}{\sigma_{f}\left(t\right)}-1\right)\left(\prod_{i\in S\setminus T'}a_{is}\prod_{i\in T'}a_{it}-\prod_{i\in S}a_{is}\prod_{i\in T}a_{it}\right) \\
&\ge0
\end{align*}
since $\sigma_{f}\left(t\right)\le\sigma_{f}\left(s\right)+1$ and
the $s$th column of $A$ is dominated by the $t$th column.

Note that when $\sigma_{f}\left(t\right)=\sigma_{f}\left(s\right)+1$
we get $T'=T$ and thus $f'=f$,$g'=g$; in this case $P_{f}=P_{g}$
so no harm is done by counting it twice.
%{\hspace*{\fill}\rule{6pt}{6pt}\smallskip}
\end{proof}

\begin{lem}
\label{lem:monotone-bijection}Let $n$ and $k$ be integers such
that $0<2k\le n$. Then there is a monotone bijection $\phi:\binom{\left[n\right]}{k}\to\binom{\left[n\right]}{n-k}$;
that is, $\phi$ maps $k$-subsets of $\left[n\right]$ to $\left(n-k\right)$-subsets
of $\left[n\right]$ and satisfies $A\subset\phi\left(A\right)$.\end{lem}
\begin{proof}
Build a bipartite graph whose sides are $L=\binom{\left[n\right]}{k}$
and $R=\binom{\left[n\right]}{n-k}$. Connect $A\in L$ to $B\in R$
iff $A\subset B$. This is a $\binom{n-k}{k}$-regular bipartite graph,
and as such it has a perfect matching (e.g., by Hall's theorem).
%{\hspace*{\fill}\rule{6pt}{6pt}\smallskip}
\end{proof}

\bigskip

We now turn to the actual calculation of the product of the terms
$c_i = \frac{1}{i} \sum_{j=n-i+1}^n \frac{j-1}{j}$. 
Denote the $i$th Harmonic and the $i$th 2-Harmonic numbers, respectively, by $H_i = 1+\frac12+\frac13+\cdots+\frac1i$ and $H_i^{(2)}=1+\frac1{2^2} + \frac1{3^2}+\cdots+\frac1{i^2}$.
Recall that the Harmonic series diverges but the 2-Harmonic series converges to $\pi^{2}/6$.

Now write $c_i = 1- \frac1i\left(\frac1n+\frac1{n-1}+\cdots+\frac1{n-i+1}\right)=1+\left(H_{n-i}-H_n\right)/i$.
\begin{lem}
$c_{1}+c_{2}+\cdots+c_{n}=n-H_{n}^{\left(2\right)}$ for all $n\in\mathbb{N}$.
\end{lem}
\begin{proof}
First we claim that $H_n^2+H_n^{(2)}=2\sum_{k=1}^nH_k/k$.
Indeed,
\begin{align*}
H_n^2
&=\left(1+\frac{1}{2}+\cdots+\frac{1}{n}\right)^2
=\sum_{j,k=1}^{n}\frac{1}{jk}
=2\sum_{1\le j\le k\le n}\frac{1}{jk}-\sum_{1\le j=k\le n}\frac{1}{jk}\\
&=2\sum_{k=1}^{n}\frac{1}{k}\sum_{j=1}^{k}\frac{1}{j}-\sum_{k=1}^{n}\frac{1}{k^{2}}
=2\sum_{k=1}^{n}\frac{H_{k}}{k}-H_{n}^{\left(2\right)}.
\end{align*}

Now for any $i\in\left[n\right]$ we have
$$
\frac{H_{n-i}-\left(H_{n}-H_{i}\right)}i
=\frac1i\sum_{k=1}^{n-i}\left(\frac1k-\frac1{i+k}\right)
=\frac1i\sum_{k=1}^{n-i}\frac{(i+k)-k}{k\left(i+k\right)}
=\sum_{k=1}^{n-i}\frac{1}{k\left(i+k\right)},
$$
so
\begin{align*}
\sum_{i=1}^{n}\frac{H_{n-i}-\left(H_{n}-H_{i}\right)}{i}
&=\sum_{i=1}^{n}\sum_{k=1}^{n-i}\frac{1}{k\left(i+k\right)}
 =\sum_{k=1}^{n}\sum_{i=1}^{n-k}\frac{1}{k\left(i+k\right)}\\
&=\sum_{k=1}^{n}\sum_{j=k+1}^{n}\frac{1}{kj}
 =\sum_{k=1}^{n}\frac{H_{n}-H_{k}}{k}\\
&=H_n^2-\sum_{k=1}^{n}\frac{H_{k}}{k},
\end{align*}
and thus
\begin{align*}
\sum_{i=1}^nc_i
&=\sum_{i=1}^n\left(1+\frac{H_{n-i}-H_n}i\right)
=n+\sum_{i=1}^n\frac{H_{n-i}-(H_n-H_i)-H_i}i\\
&=n+H_n^2-2\sum_{i=1}^n\frac{H_i}i
=n-H_n^{\left(2\right)}.
\end{align*}
%{\hspace*{\fill}\rule{6pt}{6pt}\smallskip}
\end{proof}

\begin{corollary}
$\prod_{i=1}^{n} c_i \ge\exp\left(-\left(\pi^{2}/6+o\left(1\right)\right)\right)$;
in particular, $\prod_{i=1}^{n}c_i\ge0.19$.
\end{corollary}

\begin{proof}
Use $\ln(x)\ge (x-1)/x$ for $0<x<1$ along with the fact that
$$1-\frac1n=c_1\ge c_2\ge\cdots\ge c_{n}=1-\frac{H_n}n$$
to establish that
\begin{align*}
\ln\left(\prod_{i=1}^{n}c_i\right)
&=\sum_{i=1}^n\ln\left(c_i\right)
\ge\sum_{i=1}^n\frac{c_i-1}{c_i}
\ge\sum_{i=1}^n\frac{c_i-1}{c_n}
=\frac1{c_n}\left(-n+\sum_{i=1}^n c_i\right)\\
&=\frac{-H_{n}^{\left(2\right)}}{1-H_{n}/n}\ge-\left(\frac{\pi^{2}}{6}+O\left(\frac{1}{n}\right)\right).
\end{align*}
%{\hspace*{\fill}\rule{6pt}{6pt}\smallskip}
\end{proof}

\section{Empirical Results} \label{app:empirical}
We describe our empirical results for finding PNE in random instances of Voronoi choice games.  Specifically, we investigated the likelihood that a myopic best response process would converge to a PNE on random instances of the One Way 1-D Simultaneous Voronoi Game, the 1-D Simultaneous Voronoi Game and the 2-D Simultaneous Voronoi Game.  We investigated both the maximization and minimization variants of each of these games with $n \in \{10, 100, 1000, 10000\}$ and $m \in \{2,3,4\}$.  (For the 2-D case we present results only for $n$ up to 1000.)
We note that we present limited results for the 1-D Simultaneous Voronoi Game, as we were always able to find a PNE.  This is not surprising, given Theorem~\ref{thm:simple}.  
However, it offers some weight suggesting that the worst-case time for myopic best response to converge to an equilibrium for this problem is polynomial in $n$;  this remains
an interesting open question.  

Our experimental procedure is as follows.  For each type of game, for
each value of $n$ and $m$, we generated 1000 random instances of the
game by choosing each player's points uniformly at random from the
unit torus.  For each instance, we first randomly assign each player
to one of its points.  We then repeatedly iterate through all of the
players, moving them to their best response to the current set of
points.  If in a pass through the players, none of them change
strategies, we have found a PNE.  Otherwise, we give up after 1000 such
passes.  (The number 1000 appears suitable; in our experiments, no
instances successfully converged after more than 764 passes).  If no
PNE was found, using the same set of random points, we try a new random
starting assignment of points to players, and then try another 1000
passes through the players.  If after 100 such random initial
configurations no PNE has been found, we declare that we have failed
to find a PNE for that problem instance.  (For the 2-D Simultaneous Voronoi
Game, our results are only up to 10 initial configurations;  as we shall see this
appears sufficient for our exploration.)

The program for the 2-D Simultaneous Voronoi problem was built on top
of the Delaunay triangulation program in the Computational Geometry
Algorithms Library,\footnote{\url{http://www.cgal.org/}} which
supports fast dynamic updates of a Delaunay triangulation.  By
maintaining this underlying structure, our program was able to
efficiently update the Voronoi cell areas during the myopic best
response process.

Our results are presented in Table~\ref{tab:1way1dmin}
through Table~\ref{tab:2ditrs}.  Table~\ref{tab:1way1dmin}
through Table~\ref{tab:2dmax} show in how many instances out of 1000
random instances our algorithm found a PNE, for various values
of $n$, $m$ and numbers of random initial configurations
attempted.  Table~\ref{tab:1ditrs} and Table~\ref{tab:2ditrs} give the maximum number of passes of
myopic best response that our algorithm took before successfully finding a PNE for the variants of the 1-D Simultaneous Voronoi
Game and  2-D Simultaneous Voronoi
Game respectively.

Based on this data, we propose a few conjectures.  Looking
at Table~\ref{tab:1way1dmax}, it seems possible that the probability
that a PNE exists for a random instance of the maximization variant of
the One Way 1-D Simultaneous Voronoi Game converges as $n$ grows to a constant in
$(0,1)$ for $m = 2$.  For $m \geq 3$ for the maximization variant, and for $m \geq 2$ for the minimization variant, 
Table~\ref{tab:1way1dmin} and Table~\ref{tab:1way1dmax} suggest that the probability
of a PNE exists may converge to 0, or may be difficult to find through myopic 
best response.  Both of these possibilities are consistent with our results
on the expected number of solutions.  

For the 2-D Simultaneous Voronoi Game, for both the maximization and the 
minimization version, Table~\ref{tab:2dmin} and Table~\ref{tab:2dmax} suggest
that a PNE exists for a random instance with high probability, and further
can be found by myopic best response.   

Interestingly, the expected number
of passes before finding a PNE in our experiments seems to grow like 
$\Theta(\log n)$ for the maximization variant of both the 1-D Simultaneous Voronoi Game and the 2-D Simultaneous Voronoi Game,
but seems to possibly grow like $\Theta(n)$ (or possible some other superlogarithmic
function) for the minimization variants.
Again, we are left with many interesting open questions.

\begin{table}%[h]
	\centering
		\begin{tabular}{| c | c | c | c | c | c | c |}
  \hline                       
  $m$ & $n$ & 1 attempt & 2 attempts & 5 attempts & 10 attempts & 100 attempts \\
  \hline 
  \hline
  2 & 10 & 216 & 222 & 231 & 239 & 242 \\
  \hline
  2 & 100 & 17 & 17 & 18 & 21 & 24 \\
  \hline
  2 & 1000 & 1 & 1 & 1 & 1 & 1 \\
  \hline
  2 & 10000 & 0 & 0 & 0 & 0 & 0 \\
  \hline
  3 & 10 & 38 & 40 & 45 & 46 & 46 \\
  \hline
  3 & 100 & 0 & 0 & 0 & 0 & 0 \\
  \hline
  3 & 1000 & 0 & 0 & 0 & 0 & 0 \\
  \hline
  3 & 10000 & 0 & 0 & 0 & 0 & 0 \\
  \hline
  4 & 10 & 6 & 6 & 6 & 6 & 6 \\
  \hline
  4 & 100 & 0 & 0 & 0 & 0 & 0 \\
  \hline
  4 & 1000 & 0 & 0 &  0&  0& 0 \\
  \hline
  4 & 10000 &0  &0  & 0 & 0 & 0 \\
  \hline
\end{tabular}
	\caption{The number of successes over 1000 iterations of finding a PNE for a random instance of the minimization variant of the One Way 1-D Simultaneous Voronoi Game using a given number of attempts at myopic best response with random initial strategy profiles.}
	\label{tab:1way1dmin}
\end{table}

\begin{table}%[h]
	\centering
		\begin{tabular}{| c | c | c | c | c | c | c |}
  \hline                       
  $m$ & $n$ & 1 attempt & 2 attempts & 5 attempts & 10 attempts & 100 attempts \\
  \hline 
  \hline
  2 & 10 & 616 & 633 & 652 & 663 & 668 \\
  \hline
  2 & 100 & 531 & 556 & 578 & 588 & 609 \\
  \hline
  2 & 1000 & 497 & 525 & 547 & 556 & 589 \\
  \hline
  2 & 10000 & 497 & 518 & 548 & 563 & 592 \\
  \hline
  3 & 10 & 175 & 197 & 223 & 240 & 285 \\
  \hline
  3 & 100 & 9 & 13 & 19 & 26 & 51 \\
  \hline
  3 & 1000 & 0 & 0 & 0 & 0 & 0 \\
  \hline
  3 & 10000 & 0 & 0 & 0 & 0 & 0 \\
  \hline
  4 & 10 & 35 & 44 & 56 & 65 & 112 \\
  \hline
  4 & 100 & 0 &0  & 0 &  0& 0 \\
  \hline
  4 & 1000 & 0 & 0 &0  & 0 & 0 \\
  \hline
  4 & 10000 & 0 & 0 &0  & 0 & 0 \\
  \hline
\end{tabular}
	\caption{The number of successes over 1000 iterations of finding a PNE for a random instance of the maximization variant of the One Way 1-D Simultaneous Voronoi Game using a given number of attempts at myopic best response with random initial strategy profiles.}
	\label{tab:1way1dmax}
\end{table}

% 2D Tables Start Here -------------------------------

\begin{table}%[h]
	\centering
		\begin{tabular}{| c | c | c | c | c | c |}
  \hline                       
  $m$ & $n$ & 1 attempt & 2 attempts & 5 attempts & 10 attempts \\
  \hline 
  \hline
  2 & 10 & 987 & 995 & 998 & 999 \\
  \hline
  2 & 100 & 982 & 994 & 997 & 999\\
  \hline
  2 & 1000 & 978 & 999 & 1000 & 1000\\
  \hline
  3 & 10 & 986 & 988 & 995 & 996\\
  \hline
  3 & 100 & 983 & 995 & 1000 & 1000\\
  \hline
  3 & 1000 & 967 & 1000 & 1000 & 1000\\
  \hline
  4 & 10 & 970 & 993 & 998 & 1000\\
  \hline
  4 & 100 & 969 & 1000 & 1000 & 1000\\
  \hline
  4 & 1000 & 948 & 998 & 1000 & 1000\\
  \hline
\end{tabular}
	\caption{The number of successes over 1000 iterations of finding a PNE for a random instance of the minimization variant of the 2-D Simultaneous Voronoi Game using a given number of attempts at myopic best response with random initial strategy profiles.}
	\label{tab:2dmin}
\end{table}

\begin{table}%[h]
	\centering
		\begin{tabular}{| c | c | c | c | c | c |}
  \hline                       
  $m$ & $n$ & 1 attempt & 2 attempts & 5 attempts & 10 attempts \\
  \hline 
  \hline
  2 & 10 & 992 & 997 & 999 & 999 \\
  \hline
  2 & 100 & 994 & 999 & 1000 & 1000\\
  \hline
  2 & 1000 & 991 & 1000 & 1000 & 1000\\
  \hline
  3 & 10 & 987 & 992 & 997 & 998 \\
  \hline
  3 & 100 & 990 & 999 & 1000 & 1000\\
  \hline
  3 & 1000 & 993 & 999 & 1000 & 1000\\
  \hline
  4 & 10 & 992 & 1000 & 1000 & 1000\\
  \hline
  4 & 100 & 990 & 1000 & 1000 & 1000\\
  \hline
  4 & 1000 & 993 & 1000 & 1000 & 1000\\
  \hline
\end{tabular}
	\caption{The number of successes over 1000 iterations of finding a PNE for a random instance of the maximization variant of the 2-D Simultaneous Voronoi Game using a given number of attempts at myopic best response with random initial strategy profiles.}
	\label{tab:2dmax}
\end{table}

% Max Iteration Tables Start Here -------------------------------

%\begin{table}%[h]
%	\centering
%		\begin{tabular}{| c | c | c | c | c | c | c |}
%  \hline                       
%  $n$ & min/$m=2$ & max/$m=2$ & min/$m=3$ & max/$m=3$ & min/$m=4$ & max/$m=4$ \\
%  \hline 
%  \hline
%  10 & 7 & 7 & 6 & 8 & 2 & 11 \\
%  \hline
%  100 & 36 & 17 &  & 8 &  &  \\
%  \hline
%  1000 & 74 & 20 &  &  &  \\
%  \hline
%  10000 &  & 32 &  &  &  \\
%  \hline
%\end{tabular}
%	\caption{The maximum number of passes of myopic best response before an instance successfully converged to a PNE for the various One Way 1-D Simultaneous Voronoi Games.  Those left empty had no instances which successfully converged to a PNE.}
%	\label{tab:1way1ditrs}
%\end{table}
%
%\begin{table}%[h]
%	\centering
%		\begin{tabular}{| c | c | c | c | c | c | c |}
%  \hline                       
%  $n$ & min/$m=2$ & max/$m=2$ & min/$m=3$ & max/$m=3$ & min/$m=4$ & max/$m=4$ \\
%  \hline 
%  \hline
%  10 & 8 & 7 & 7 & 8 & 4 & 6 \\
%  \hline
%  100 & 20 & 10 &  &  &  &  \\
%  \hline
%  1000 & 54 & 12 &  &  &  \\
%  \hline
%  10000 & 73 & 15 &  &  &  \\
%  \hline
%\end{tabular}
%	\caption{The maximum number of passes of myopic best response before an instance successfully converged to a PNE for the various 1-D Simultaneous Voronoi Games.  Those left empty had no instances which successfully converged to a PNE.}
%	\label{tab:1ditrs}
%\end{table}

\begin{table}%[h]
	\centering
		\begin{tabular}{| c | c | c | c | c | c | c |}
  \hline                       
  $n$ & min/$m=2$ & min/$m=3$ & min/$m=4$ & max/$m=2$ & max/$m=3$ & max/$m=4$ \\
  \hline 
  \hline
  10 & 6 & 8 & 11 & 5 & 6 & 8  \\
  \hline
  100 & 12 & 22 & 37 & 8 & 9 & 12  \\
  \hline
  1000 & 42 & 87 & 173 & 11 & 12 & 16  \\
  \hline
  10000 & 117 & 361 & 764 & 12 & 15 & 19  \\
  \hline
\end{tabular}
	\caption{The maximum number of passes of myopic best response before an instance converged to a PNE for various 1-D Simultaneous Voronoi Games.}
	\label{tab:1ditrs}
\end{table}

\begin{table}%[h]
	\centering
		\begin{tabular}{| c | c | c | c | c | c | c |}
  \hline                       
  $n$ & min/$m=2$ & min/$m=3$ & min/$m=4$ & max/$m=2$ & max/$m=3$ & max/$m=4$ \\
  \hline 
  \hline
  10 & 7 & 8 & 13 & 7 & 8 & 10 \\
  \hline
  100 & 15 & 14 & 19 & 13 & 13 & 16 \\
  \hline
  1000 & 29 & 30 & 57 & 13 & 14 & 19 \\
  \hline
\end{tabular}
	\caption{The maximum number of passes of myopic best response before an instance successfully converged to a PNE for various 2-D Simultaneous Voronoi Games.}
	\label{tab:2ditrs}
\end{table}

\end{document}